\documentclass{sig-alternate}
\usepackage[english]{babel}
\usepackage{blindtext}
\usepackage{etoolbox}
\makeatletter
\patchcmd{\maketitle}{\@copyrightspace}{}{}{}
\makeatother

\usepackage{times,amsmath,epsfig}
\usepackage{graphicx}
\usepackage{balance}  
\usepackage{subfig}
\usepackage{xspace}
\usepackage{amssymb}
\usepackage{algorithm}
\usepackage{algpseudocode}
\algnewcommand\algorithmicinput{\textbf{Input:}}
\algnewcommand\Input{\item[\algorithmicinput]}
\algnewcommand\algorithmicoutput{\textbf{Output:}}
\algnewcommand\Output{\item[\algorithmicoutput]}

\usepackage{graphicx}
\usepackage{cite}
\usepackage{url}
\usepackage{alltt}

\newtheorem{definition}{Definition}
\newtheorem{theorem}{Theorem}
\newtheorem{lemma}{Lemma}
\newtheorem{proposition}{Proposition}

\newcommand{\SmallSpace}{\vspace*{-1.5ex}}

\newcommand{\baseline}{\textsf{\small Baseline}\xspace}
\newcommand{\privqt}{\textsf{\small PrivQT}\xspace}
\newcommand{\privthr}{\textsf{\small PrivTHR}\xspace}
\newcommand{\privthrem}{\textsf{\small PrivTHR$_{EM}$}\xspace}

\newcommand{\eat}[1]{}

\title{WaveCluster with Differential Privacy}

\author{
{Ling Chen{\small $~^{\#1}$}, Ting Yu{\small $~^{\#1,2}$}, Rada Chirkova{\small $~^{\#1}$} }%
\vspace{1.6mm}\\
\fontsize{10}{10}\selectfont\itshape
$^{\#1}$\,Department of Computer Science, North Carolina State University, Raleigh, USA\\
\fontsize{10}{10}\selectfont\itshape
$^{\#2}$\,Qatar Computing Research Institute, Doha, Qatar\\
\fontsize{9}{9}\selectfont\ttfamily\upshape
$^{1}$\,lchen10@ncsu.edu, $^{1,2}$\,tyu@\{ncsu.edu,qf.org.qa\}, $^{1}$\,rychirko@ncsu.edu
}

\begin{document}
\maketitle

\begin{abstract}
WaveCluster is an important family of grid-based clustering algorithms that are capable of finding clusters of arbitrary shapes. In this paper, we investigate techniques to perform WaveCluster while ensuring differential privacy. 
Our goal is to develop a general technique for achieving differential privacy on WaveCluster that accommodates different wavelet transforms.
We show that straightforward techniques based on synthetic data generation and introduction of random noise when quantizing the data, though generally preserving the distribution of data, often introduce too much noise to preserve useful clusters. We then propose two optimized techniques, \privthr and \privthrem, which can significantly reduce data distortion during two key steps of WaveCluster: the quantization step and the significant grid identification step. We conduct extensive experiments based on four datasets that are particularly interesting in the context of clustering, and show that \privthr and \privthrem achieve high utility when privacy budgets are properly allocated. 
\end{abstract}

\section{Introduction}\label{section-introduction}
Clustering is an important class of data analysis that has been extensively applied in a variety of fields, such as identifying different groups of customers in marketing and grouping homologous gene sequences in biology research~\cite{datamininghan}. Clustering results allow data analysts to gain valuable insights into data distribution when it is challenging to make hypotheses on raw data. Among various clustering techniques, a grid-based clustering algorithm called WaveCluster~\cite{wavecluster,waveclusterjournal} is famous for detecting clusters of arbitrary shapes. WaveCluster relies on wavelet transforms, a family of convolutions with appropriate kernel functions, to convert data into a transformed space, where the natural clusters in the data become more distinguishable.

In many data-analysis scenarios, when the data being analyzed contains personal information and the result of the analysis needs to be shared with the public or untrusted third parties, sensitive private information may be leaked, e.g., whether certain personal information is stored in a database or has contributed to the analysis. Consider the databases \textit{A} and \textit{B} in Figure~\ref{fig:privacy-breach}. These two databases have two attributes, \textit{Monthly Income} and \textit{Monthly Living Expenses}, and the records differ only in one record, \textit{u}. Without \textit{u}'s participation in database \textit{A}, WaveCluster identifies two separate clusters, marked by \textit{blue} and \textit{red}, respectively. With \textit{u}'s participation, WaveCluster identifies only one cluster marked by color \textit{blue} from database \textit{B}. Therefore, merely from the number of clusters returned (rather than which data points belong to which cluster), an adversary may infer a user's participation. Due to such potential leak of private information, data holders may be reluctant to share the original data or data-analysis results with each other or with the public.

\begin{figure}
\centering
\subfloat[]{\includegraphics[scale=0.36,clip]{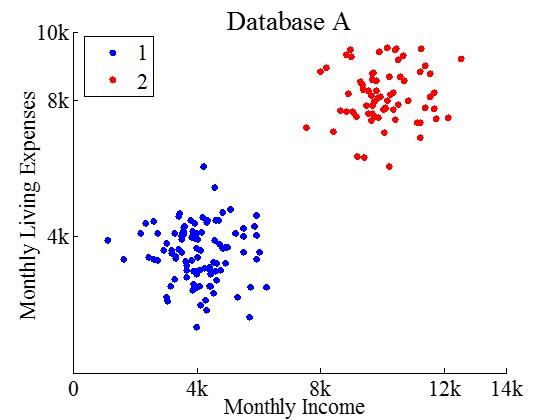} } 
\hspace*{-3ex}
\subfloat[]{\includegraphics[scale=0.37,clip]{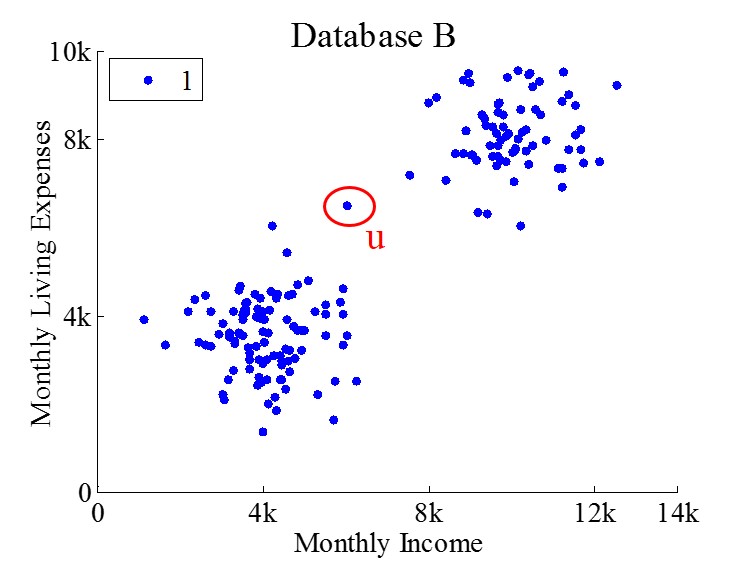} }
\vspace*{-0.5ex}
\caption{\label{fig:privacy-breach}Example of personal privacy breach in cluster analysis.} 
\vspace*{-4ex}
\end{figure}

In this paper, we develop techniques to perform WaveCluster with differential privacy~\cite{Dwork_survey,calibrating}. 
Differential privacy provides a provable strong privacy guarantee that the output of a computation is insensitive to any particular individual. In other words, based on the output, an adversary has limited ability to make inference about whether an individual is present or absent in the dataset. Differential privacy is often achieved by the perturbation of randomized algorithms,  and the privacy level is controlled by a parameter $\epsilon$ called ``privacy budget''. Intuitively, the privacy protection via differential privacy grows stronger as $\epsilon$ grows smaller.

WaveCluster provides a framework that allows any kind of wavelet transform to be plugged in for data transformation, such as the Haar transform~\cite{haar} and Biorthogonal transform~\cite{bior}.
There are various wavelet transforms that are suitable for different types of applications, such as image compression and signal processing~\cite{waveletreview}.
Plugged in different wavelet transforms, WaveCluster can leverage different properties of the data, such as frequency and location, for finding the dense regions as clusters.
Thus, in this paper, we aim to develop a general technique for achieving differential privacy on WaveCluster that accommodates different wavelet transforms.

\begin{figure}
\centering
\subfloat[Original]{\includegraphics[scale=0.33,clip]{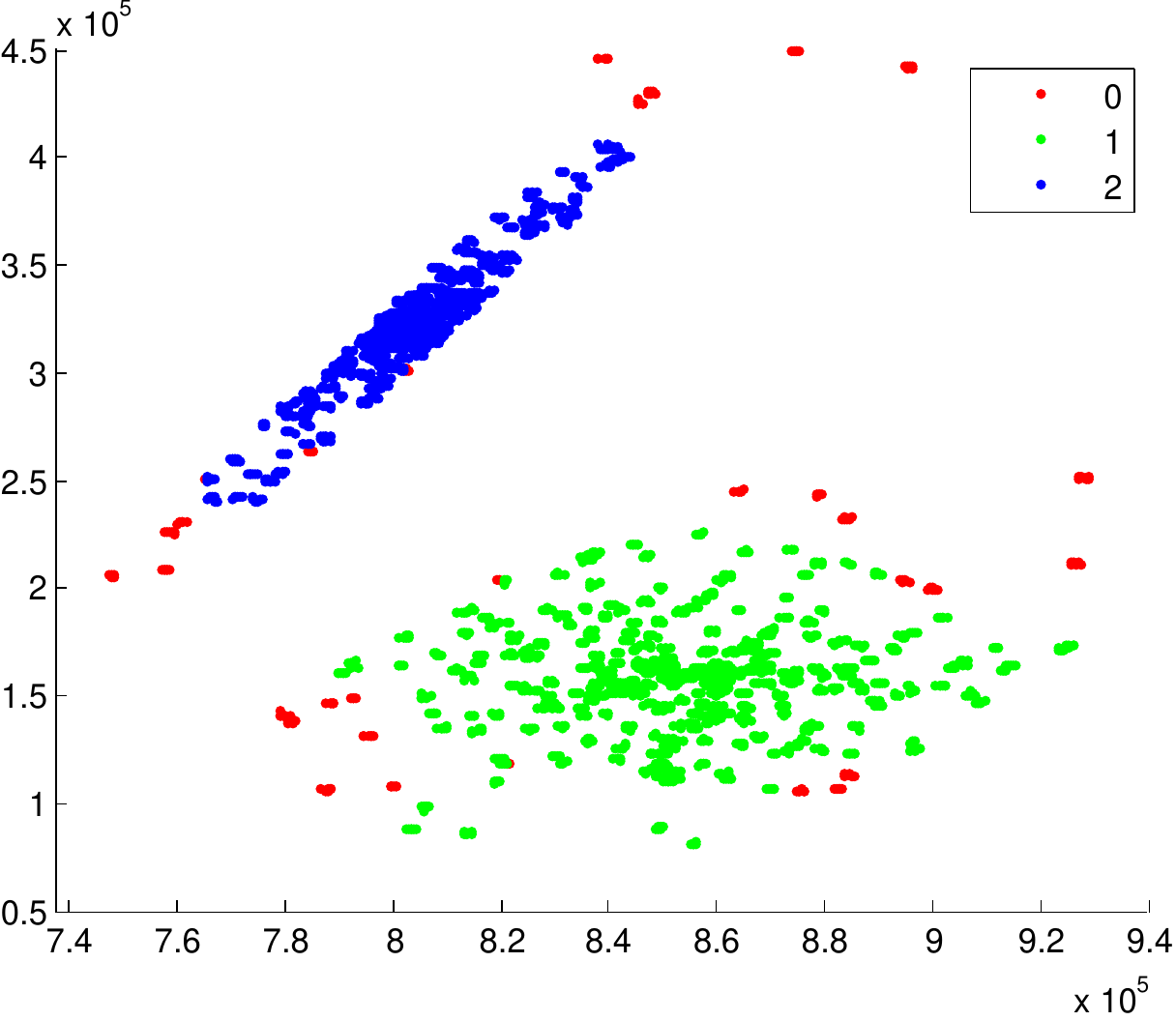} } 
\hspace*{1ex}
\subfloat[\baseline]{\includegraphics[scale=0.33,clip]{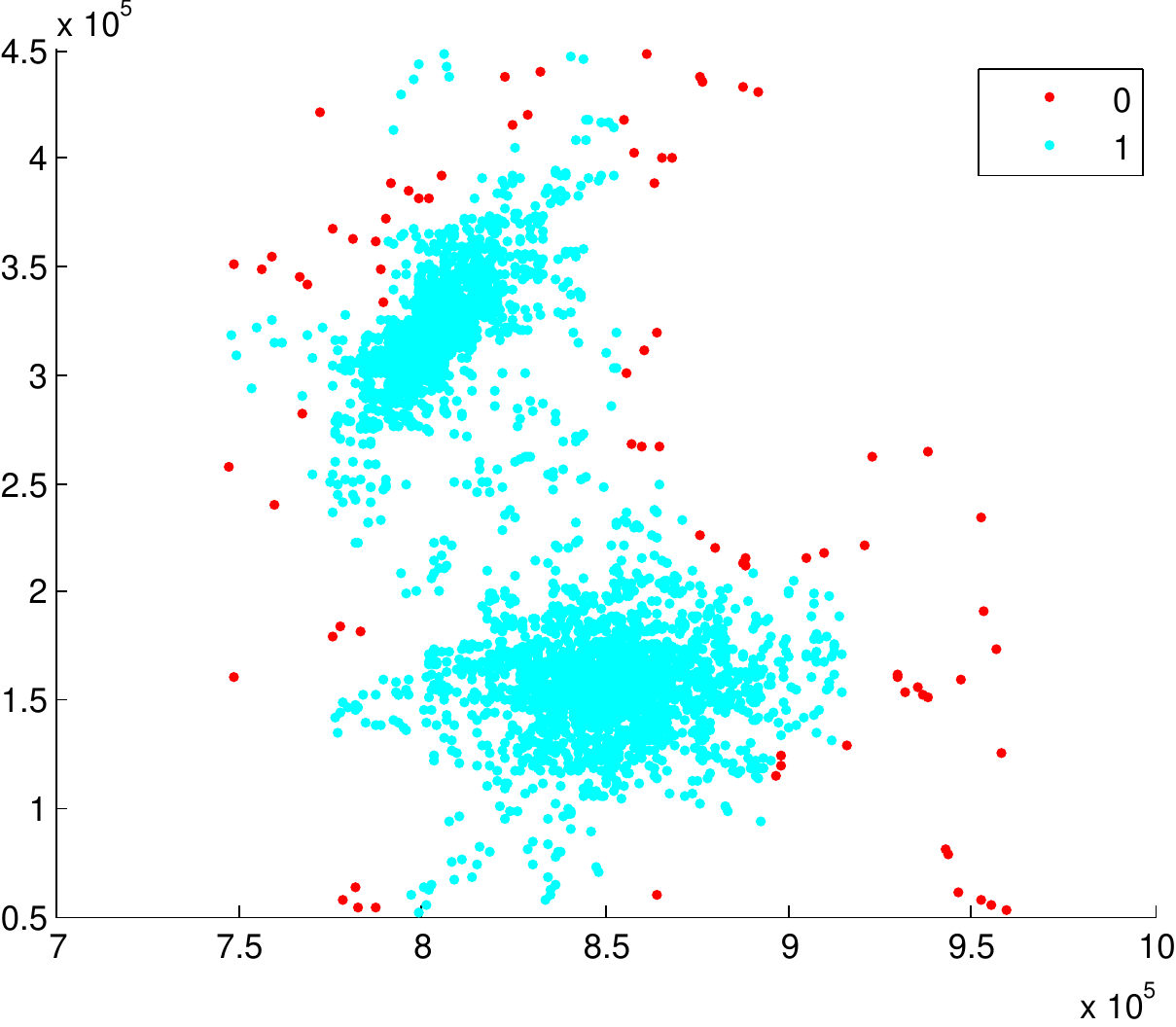} }
\vspace*{-0.5ex}
\caption{\label{fig:baselinediff}Inaccurate clustering result produced by \baseline. (a) shows the WaveCluster results on the original data
and (b) shows the WaveCluster results of \baseline,
which leverages the adaptive-grid~\cite{geospatial} approach to generate the synthetic data. Points in different clusters are shown in different colors,
and the points marked by \textit{red} are considered as noises that do not form a cluster.} 
\vspace*{-4ex}
\end{figure}

We first consider a general technique, \baseline, that adapts existing differentially private data-publishing techniques to WaveCluster through synthetic data generation. Specifically, we could generate synthetic data based on any data model of the original data that is published through differential privacy, and then apply WaveCluster using any wavelet transform over the synthetic data. \baseline seems particularly promising as many effective differentially private data-publishing techniques have been proposed in the literature, all of which strive to preserve some important properties of the original data. Therefore, hopefully the ``shape'' of the original data is also preserved in the synthetic data, and consequently could be discovered by WaveCluster. 
Unfortunately, as we will show later in the paper, this synthetic data-generation technique often cannot produce accurate results. Differentially private data-publishing techniques such as spatial decompositions~\cite{spatial}, adaptive-grid~\cite{geospatial}, and Privelet~\cite{privelet}, output noisy descriptions of the data distribution and often contain negative counts for sparse partitions due to random noise. These negative counts do not affect the accuracy of large range queries (which is often one of the main utility measures in private data publishing) since zero-mean noise distribution smoothes the effect of negative counts. However, negative counts cannot be smoothed away in the synthesized dataset, which are typically set to zero counts. 
Figure~\ref{fig:baselinediff} shows an example of inaccurate clustering results produced by \baseline using adaptive-grid~\cite{geospatial},
As we can see, the synthetic data generated in \baseline significantly distorts the data distribution, 
causing two clusters to be merged as one and reducing the accuracy of the WaveCluster results.

Motivated by the above challenge, we propose three techniques that enforce differential privacy on the key steps of WaveCluster, rather than relying on synthetic data generation. WaveCluster accepts as input a set of data points in a multi-dimensional space, and consists of the following main steps. First, in the quantization step WaveCluster quantizes the multi-dimensional space by dividing the space into grids, and computes the count of the data points in each grid. These counts of grids form a count matrix $M$. Second, in the wavelet transform step WaveCluster applies a wavelet transform on the count matrix $M$ to obtain the approximation of the multi-dimensional space. Third, in the significant grid identification step WaveCluster identifies significant grids based on the pre-defined density threshold. Fourth, in the cluster identification step WaveCluster outputs as clusters the connected components from these significant grids~\cite{hor88}. 
To enforce differential privacy on WaveCluster, we first propose a technique, \privqt, that introduces Laplacian noise to the quantization step. However, such straightforward privacy enforcement cannot produce usable private WaveCluster results, since the noise introduced in this step significantly distorts the density threshold for identifying significant grids. 
To address this issue, we further propose two techniques, \privthr and \privthrem, which enforce differential privacy on both the quantization step and the significant grid identification step. These two techniques differ in how to determine the noisy density threshold. We show that by allocating appropriate budgets in these two steps, both techniques can achieve differential privacy with significantly improved utility.

Traditionally, the effectiveness of WaveCluster is evaluated through visual inspection by human experts (i.e., visually determining whether the discovered clusters match those reflected in the user's mind)~\cite{wavecluster, waveclusterjournal}. Unfortunately, visual inspection is inappropriate to assess the utility of differentially private WaveCluster. Visual inspection is not quantitative, and thus it is hard to systematically compare the impact of different techniques through visual inspection. Generally, researchers use quantitative measures to assess the utility of differentially private results, such as relative or absolute errors for range queries and prediction accuracy for classification. But there is no existing utility measures for density-based clustering algorithms with differential privacy.

To mitigate this problem, in this paper we propose two types of utility measures. 
The first is to measure the dissimilarity between true and private WaveCluster results by measuring the differences of significant grids and clusters, which correspond to the outputs of the two key steps (the significant grid identification and the cluster identification) in WaveCluster.
To more intuitively understand the usefulness of discovered clusters, our second utility measure considers one concrete application of cluster analysis, i.e., to build a classifier based on discovered clusters, and then use that classifier to predict future data. Therefore the prediction accuracy of the classifier from one aspect reflects the actual utility of private WaveCluster.

To evaluate the proposed techniques, our experiments use four datasets containing different data shapes that are particularly interesting in the context of clustering~\cite{datasetlink, Gowalla}.
Our results show that \privthr and \privthrem achieve high utility for both types of utility measures, and are superior to \baseline and \privqt.

\section{Related Work}\label{section-relatedwork}
The syntactic approaches for privacy preserving clustering~\cite{clusteringsurvey} is to output $k$-anonymous clusters.
Friedman et al.~\cite{kanonymityforDM} presented an algorithm to output $k$-anonymous clusters by using minimum spanning tree.
Karakasidis et al.~\cite{privateblocking} created $k$-anonymous clusters by merging clusters so that each cluster contains at least $k$ key values of the records.
Fung et al.~\cite{datapublishingforclustering} proposed an approach that converts the anonymity problem for cluster analysis to the counterpart problem for classification analysis.
Aggarwal et al.~\cite{rgatherclustering} proposed a perturbation method called $r$-gather clustering, 
which releases the cluster centers, together
with their sizes, radiuses, and a set of associated sensitive values.
However, these approaches only satisfy syntactic privacy notions such as k-anonymity, 
and cannot provide formal guarantees of privacy as differential privacy.

In this work, our goal is to perform WaveCluster under differential privacy.
The focus of initial work on differential privacy~\cite{Dwork_survey,calibrating,LearnPrivately,Dwork_robust_stat,private_corset} concerned the theoretical proof of its feasibility on various data analysis tasks, e.g., histogram and logistic regression. 

More recent work has focused on practical applications of differential privacy for privacy-preserving data publishing. 
An approach proposed by Barak et al.~\cite{Barak} encoded marginals with Fourier coefficients and then added noise to the released coefficients.
Hay et al.~\cite{Hay} exploited consistency constraints to reduce noise for histogram counts.  
Xiao et al.~\cite{privelet} proposed \textit{Privelet}, which uses wavelet transforms to reduce noise for histogram counts.
Cormode et al.~\cite{spatial} indexed data by \textit{kd}-trees and \textit{quad}-trees, developing 
effective budget allocation strategies for building the noisy trees and obtaining noisy counts for the tree nodes. 
Qardaji et al.~\cite{geospatial} proposed uniform-grid and adaptive-grid methods to derive appropriate partition granularity in differentially private synopsis publishing.
Xu et al.~\cite{histogram} proposed the \textit{NoiseFirst} and \textit{StructureFirst} techniques for constructing optimal noisy histograms, using dynamic programming and Exponential mechanism. 
These data publishing techniques are specifically crafted for answering range queries.
Unfortunately, synthesizing the dataset and applying WaveCluster on top of it often render WaveCluster results useless, since these differentially private data publishing techniques do not capture the essence of WaveCluster and introduce too much unnecessary noise for WaveCluster. 

Another important line of prior work focuses on integrating differential privacy into other practical data analysis tasks, such as regression analysis, model fitting, classification and etc. 
Chaudhuri et al.~\cite{regression_nips} proposed a differentially private regularized logistic regression algorithm that balances privacy with learnability. 
Zhang et al.~\cite{functionalmechanism} proposed a differentially private approach for logistic and linear regressions that involve perturbing the objective function of the regression model, rather than simply introducing noise into the results. 
Friedman et al.~\cite{datamining1} incorporated differential privacy into several types of decision trees and subsequently demonstrated the tradeoff among privacy, accuracy and sample size.
Using decision trees as an example application, Mohammed et al.~\cite{datamining2} investigated a generalization-based algorithm for achieving differential privacy for classification problems.

Differentially private cluster analysis has also be studied in prior work. 
Zhang et al.~\cite{genetic} proposed differentially private model fitting based on genetic algorithms, with applications to \textit{k}-means clustering. 
McSherry~\cite{pinq} introduced the PINQ framework, which has been applied to achieve differential privacy for \textit{k}-means clustering using an iterative algorithm~\cite{damson}.
Nissim et al.\cite{smooth} proposed the sample-aggregate framework that calibrates the noise magnitude according to the smooth sensitivity of a function.
They showed that their framework can be applied to \textit{k}-means clustering under the assumption that the dataset is well-separated. 
These research efforts primarily focus on centroid-based clustering, such as \textit{k}-means, 
that is most suited for separating convex clusters and presents insufficient spatial information to detect clusters with complex shapes, e.g. concave shapes.
In contrast to these research efforts, we propose techniques that enforce differential privacy on WaveCluster, which is not restricted to well-separated datasets, and can detect clusters with arbitrary shapes.

\section{Preliminaries}\label{section-preliminaries}
In this section, we first present the background of differential privacy. Then we describe the WaveCluster algorithm followed by our problem statement.

\subsection{Differential Privacy}
\label{subsec:diffprivacy}
Differential privacy~\cite{Dwork_survey} is a recent privacy model, which guarantees that an adversary cannot infer an individual's presence in a dataset from the randomized output, despite having knowledge of all remaining individuals in the dataset.

\begin{definition}
\textbf{($\epsilon$-differential privacy)}: 
Given any pair of neighboring databases $D$ and $D'$ that differ only in one individual record, a randomized algorithm $A$ is $\epsilon$-differentially private iff for any $S \subseteq Range(A)$:
\vspace*{-2ex}
$$Pr[A(D) \in S] \leq Pr[A(D') \in S] * e^{\epsilon}$$
\vspace*{-4.5ex}
\end{definition} 

The parameter $\epsilon$ indicates the level of privacy. 
Smaller $\epsilon$ provides stronger privacy. 
When $\epsilon$ is very small, $e^\epsilon$ $\approx$ 1+ $\epsilon$. 
Since the value of $\epsilon$ directly affects the level of privacy, we refer to it as the \emph{privacy budget}.
Appropriate allocation of the privacy budget for a computational process is important for reaching a favorable trade-off between privacy and utility.
The most common strategy to achieve $\epsilon$-differential privacy is to add noise to the output of a function.
The magnitude of introduced noise is calibrated by the privacy budget $\epsilon$ and the sensitivity of the query function. 
The sensitivity of a query function is defined as the maximum difference between the outputs of the query function on any pair of neighboring databases.:

\vspace*{-2ex}
$$
\Delta f = \max_{D,D'}\parallel f(D) - f(D') \parallel_1
$$
\vspace*{-2ex}

There are two common approaches for achieving $\epsilon$-differential privacy: Laplace mechanism~\cite{calibrating} and Exponential mechanism~\cite{exponential}. 

{\bf Laplace Mechanism}: The output of a query function $f$ is perturbed by adding noise from the Laplace distribution with probability density function $f(x|b) = \frac{1}{2b}\exp(-\frac{|x|}{b})$, $b = \frac{\Delta f}{\epsilon}$. The following randomized mechanism $A_l$ satisfies  $\epsilon$-differential privacy:
\vspace*{-2ex}
$$
\mathcal{A}_l(D) = f(D) + Lap(\frac{\Delta f}{\epsilon})
$$
\vspace*{-2ex}

{\bf Exponential Mechanism}: This mechanism returns an output that is close to the optimum, with respect to a quality function.
A quality function $q(D, r)$ assigns a score to all possible outputs $r \in R$, where $R$ is the output range of $f$, and better outputs receive higher scores.
A randomized mechanism $A_e$ that outputs $r \in R$ with probability
\vspace*{-2ex}
$$
Pr[\mathcal{A}_e(D) = r] \propto exp(\frac{\epsilon q(D,r)}{2S(q)} )
$$
\vspace*{-2ex}
\\
satisfies $\epsilon$-differential privacy, where $S(q)$ is the sensitivity of the quality function. 

Differential privacy has two properties: sequential composition and parallel composition. 
Sequential composition is that given $n$ independent randomized mechanisms $A_1,A_2,\ldots,A_n$ where $A_i$ ($1\leq i \leq n$) satisfies $\epsilon_i$-differential privacy, a sequence of $A_i$ over the dataset $D$ satisfies $\epsilon$-differential privacy, where $\epsilon = \sum_1^n(\epsilon_i)$.
Parallel composition is that given $n$ independent randomized mechanisms $A_1,A_2,\ldots,A_n$ where $A_i$ ($1\leq i \leq n$) satisfies $\epsilon$-differential privacy, a sequence of $A_i$ over a set of disjoint data sets $D_i$ satisfies $\epsilon$-differential privacy.

\begin{figure}
\centering\includegraphics[scale=0.369,clip]{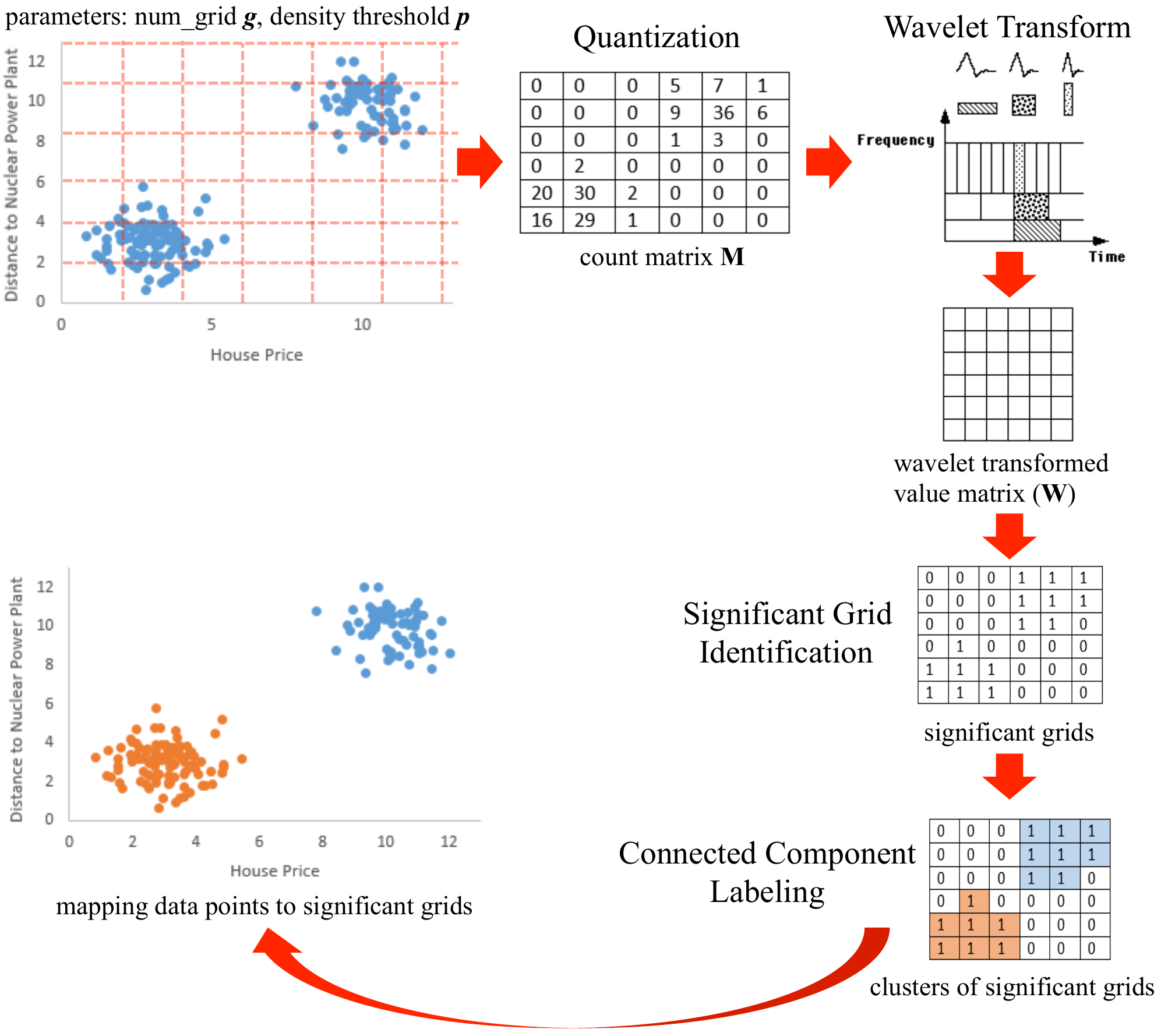} 
\vspace*{-1.5ex}
\caption{\label{fig:wavecluster}Illustration of WaveCluster.} 
\vspace*{-4ex}
\end{figure}

\subsection{WaveCluster}
\label{subsec:wavecluster}
WaveCluster is an algorithm developed by Sheikholeslami et al.~\cite{wavecluster,waveclusterjournal} for the purpose of clustering spatial data. 
It works by using a wavelet transform to detect the boundaries between clusters. 
A wavelet transform allows the algorithm to distinguish between areas of high contrast (high frequency components) and areas of low contrast (low frequency components). 
The motivation behind this distinction is that within a cluster there should be low contrast and between clusters there should be an area of high contrast (the border). 
WaveCluster has the following steps as shown in Figure~\ref{fig:wavecluster}:

\textbf{Quantization}:	Quantize the feature space into grids of a specified size, creating a count matrix $M$.

\textbf{Wavelet Transform}:	Apply a wavelet transform to the count matrix $M$, such as Haar transform~\cite{haar} and Biorthogonal transform~\cite{bior}, and decompose $M$ to the average subband that gives the approximation of the count matrix and the detail subband that has the information about the boundaries of clusters. 
We refer to the average subband as the wavelet-transformed-value matrix ($W$).

\textbf{Significant Grid Identification}: Identify the significant grids from the average subband $W$. 
WaveCluster constructs a sorted list $L$ of the positive wavelet transformed values obtained from $W$ and compute the $p$th percentile of the values in $L$. The values that are below the $p$th percentile of $L$ are non-significant values. Their corresponding grids are considered as non-significant grids and the data points in the non-significant grids are considered as noise.

\textbf{Cluster Identification}: Identify clusters from the significant grids using connected component labeling algorithm~\cite{hor88} (two grids are connected if they are adjacent), map the clusters back to the original multi-dimensional space, and label the data points based on which cluster the data points reside in.

In WaveCluster, users need to specify four parameters: 

	\textbf{num\_grid ($g_1,g_2,\ldots,g_n$)}: the number of grids that the $n$-\\dimensional space is partitioned into along each dimension.
	For the brevity of description, we simply use $g$ to refer to the partitions of the $n$-dimensional space ($g_1,g_2,\ldots,g_n$). 
	This parameter controls the scaling of quantization. 
	Inappropriate scaling can cause problems of over-quantization and under-quantization, affecting the accuracy of clustering~\cite{waveclusterjournal}.
	
	\textbf{density threshold ($p$)}: a percentage value $p$ that specifies $p$\% of the values in $L$ are non-significant values. For ease of presentation, we use $k = (1-p)|L|$ to represent the top $k$ values in $L$ and their corresponding grids are considered as significant grids.

	\textbf{level}: a wavelet decomposition level, which indicates how many times a wavelet transform is applied. The larger the level is, the more approximate the result is. 
	In our techniques, we set level to 1 since a smaller level value provides more accurate results~\cite{waveclusterjournal}.

	\textbf{wavelet}: a wavelet transform to be applied. Haar transform~\cite{haar} is one of the simplest wavelet transforms and widely used, which is computed by iterating difference and averaging between odd and even samples of a signal (or a sequence of data points). 
Other commonly used wavelet transforms include Biorthogonal transform~\cite{bior}, Daubechies transform~\cite{Daubechies}, and so on.

\textbf{Motivating Scenario.}
Consider a scenario with two participants: the data owner (e.g. hospitals) and the querier (e.g. data miner). 
The data owner holds raw data and has the legal obligation to protect individuals' privacy while the querier is eager to obtain cluster analysis results for further exploration. 
The goal of our work is to enable the data owner to release cluster analysis results using WaveCluster while not compromising the privacy of any individual who contributes to the raw data. 
The data owner has a good knowledge of the raw data and it is not difficult for her to pick the appropriate parameters (e.g. num\_grid, density threshold, and wavelet) for non-private WaveCluster.
For example, the data owner may draw from her past experience on similar data to determine the appropriate parameters for the current dataset. 
The parameters picked for the non-private setting are directly used for the private setting, 
and thus the data owner does not need to infer another set of parameters for the private setting.

\textbf{Problem Statement.}
Given a raw data set $D$, appropriate WaveCluster parameters for $D$ and a privacy budget $\epsilon$, 
our goal is to investigate an effective approach $A$ such that $A$ (1) satisfies $\epsilon$-differential privacy, 
and (2) achieves high utility of the private WaveCluster results 
with regard to the utility metrics $U$.

\section{Approaches}\label{section-approach}
In this section, we present four techniques for achieving differential privacy on WaveCluster.
We first describe the \baseline technique that achieves differential privacy through synthetic data generation.
We then describe three techniques that enforces differential privacy on the key steps of WaveCluster.

\subsection{Baseline Approach (\baseline)}
A straightforward technique to achieve differential privacy on WaveCluster is as follows: (1) adapt an existing $\epsilon$-differential privacy preserving data publishing method to get the noisy description of the data distribution in some fashion, such as a set of contingency tables or a spatial decomposition tree~\cite{histogram,privelet,spatial,geospatial}; (2) generate a synthetic dataset according to the noisy description; (3) apply WaveCluster on the synthetic dataset. 
We refer to this technique as \baseline, and its pseudocode is shown in Algorithm~\ref{alg:baseline}. 

\begin{algorithm}
\caption{\baseline}\label{alg:baseline}
\begin{algorithmic}[1]
\Input Dataset $D$, num\_grid $g$, density threshold $p$, wavelet transform $w$, differential privacy budget $\epsilon$
\Output A set of differentially private clusters
\Procedure{$Baseline$}{$D,g,p,w,\epsilon$}
   \State $D'$ = DiffPrivPublishing($D,\epsilon$)
	 \State $M'$ = Quantization($D',g$)
	 \State $W'$ = WaveletTransform($M'$,$w$)
   \State $L'$ = ConvertToPosSortedArray($W'$)
   \State $k'$ = $(1-p)|L'|$
   \State $d'$ = Top($k'$,$L'$)
   \State \textbf{return} ConnCompLabel($W',d'$)
\EndProcedure
\end{algorithmic}
\end{algorithm}

\baseline first leverages a $\epsilon$-differential privacy preserving data publishing method to 
obtain a noisy dataset $D'$ (Line 2)
and partitions $D'$ based on the number of grids $g$ to obtain the noisy count matrix $M'$ (Line 3).
\baseline then applies a wavelet transform on $M'$ to obtain $W'$ (Line 4).
$W'$ is then turned into a list $L'$ that keeps only positive values and the values in $L'$ is sorted in ascending order (Line 5).
With $L'$, $k'$ is computed based on the specified density threshold $p$ and the size of $L'$ (Line 6).
Finally, \baseline obtains $d'$ as the top $k'$th value in $L'$ (Line 7), where any value in $L'$ greater than $d'$ is considered as a significant value,
and applies the connected component labeling algorithm to identify clusters of significant grids (Line 8).

\textbf{Discussion.}
\baseline achieves differential privacy on WaveCluster through the achievement of differential privacy on data publishing. 
However, it does not produce accurate WaveCluster results in most cases.
The adapted $\epsilon$-differential privacy preserving data publishing method is designed for answering range queries. 
The noisy descriptions of the data distribution generated by the method may contain negative counts for certain partitions since the noise distribution is Laplacian with zero mean. 
These negative counts do not affect the range query accuracy too much since zero-mean noise distribution smooths the effect of noise. 
For example, a partition $p_1$ has the true count of 2 and the noisy count of -2, whose noise is canceled by another partition $p_2$ having the true count of 10 and the noisy count of 14 when both $p_1$ and $p_2$ are included in a range query. 
In particular, when the range query spreads large range of a dataset, a single partition with a noisy negative count does not affect its accuracy too much. 
However, when the method is used for generating a synthetic dataset, the noisy negative counts are reset as zero counts, causing the data distribution to change radically on the whole and further leading to the severe deviation in differentially private WaveCluster results.

\subsection{Private Quantization (\privqt)}
To address the challenge faced by \baseline, we propose techniques that enforce differential privacy on the key steps of WaveCluster.
Our first approach, called Private Quantization (\privqt), introduces independent Laplacian noise in the quantization step to achieve differential privacy. 
In the quantization step, the data is divided into grids and the count matrix $M$ is computed.
To ensure differential privacy in this step, we rely on the Laplace mechanism that introduces independent Laplacian noise to $M$. 
Clearly, if we change one individual in the input data, such as adding, removing or modifying an individual, there is at most one change in one entry of $M$.
According to the parallel composition property of differential privacy, the noise amount introduced to each grid is $Lap(\frac{1}{\epsilon})$, given a privacy budget $\epsilon$.
Since the following steps of WaveCluster are carried on using the differentially private count matrix $M'$, the clusters derived from these steps are also differentially private. 
Algorithm~\ref{alg:privqt} shows the pseudocode of \privqt. 
Except from the first step that introduces independent Laplacian noise to $M$ (Line 2), the other steps (Lines 3-7) are the same as \baseline.

\begin{algorithm}
\caption{\privqt}\label{alg:privqt}
\begin{algorithmic}[1]
\Input Dataset $D$, num\_grid $g$, density threshold $p$, wavelet transform $w$, differential privacy budget $\epsilon$
\Output A set of differentially private clusters
\Procedure{$PrivQT$}{$D,g,p,w,\epsilon$}
   \State $M'$ = PrivQuantization($D,g,\epsilon$)
	 \State $W'$ = WaveletTransform($M'$,$w$)
   \State $L'$ = ConvertToPosSortedArray($W'$)
   \State $k'$ = $(1-p)|L'|$
   \State $d'$ = Top($k'$,$L'$)
   \State \textbf{return} ConnCompLabel($W',d'$)
\EndProcedure
\end{algorithmic}
\end{algorithm}

Selecting the appropriate grid size (reflected by the parameter num\_grid $g$) in the quantization step strongly affects the accuracy of WaveCluster results~\cite{waveclusterjournal}, and also the differentially private \\WaveCluster results.
A small grid size (small $g$) causes more data points to fall into each grid and thus the count of data points for each grid becomes larger, which makes the count matrix $M$ resistant to Laplacian noise. 
However, the small grid size is not helpful for WaveCluster to detect clusters with accurate shapes and renders the results less useful. 
On the other hand, although posing a larger grid size on the data captures the density distribution of the data more clearly, it makes each grid's count too small and thus become sensitive to Laplacian noise, which dramatically affects the identification of significant grids and further the shapes of clusters. 
Our empirical results show that only when an appropriate grid size is given, differentially private WaveCluster results maintains high utility.

\textbf{Discussion.}
Although \privqt achieves differential privacy on the WaveCluster results, the noisy count matrix $M'$ and its resulting noisy $L'$ are significantly distorted and consequently the clustering results.
The reason is as follows. 
Given a specified percentage value $p$, \privqt computes $k'$ from the positive values in $W'$, where $W'$ is derived from $M'$, which is perturbed by Laplacian noise.
Laplacian distribution is symmetric and has zero-mean.
According to its randomness, approximately half of the zero-count grids become noisy positive-count grids due to positive noise while the remaining ones are turned into noisy negative-count grids due to negative noise.
These noisy positive-count grids may cause their corresponding wavelet transformed values in $W'$ to become positive (depending on the targeted wavelet transform), which will inappropriately participate in the computation of $k'$ and further distorts $k'$. 
Due to the dominating errors introduced by approximately half of zero-count grids becoming noisy positive-count grids, our empirical results show that the utility of private WaveCluster results by \privqt improves marginally even for a large privacy budget.

\subsection{Private Quantization with Refined Noisy Density Threshold (\privthr)}
The limitation of \privqt lies in the severe distortion of $k'$ by Laplacian noise introduced into count matrix $M'$.
To mitigate the distortion, we propose a technique, \privthr, which prunes a portion of noisy positive values in $W'$ to refine the computation of $k'$. 
Algorithm~\ref{alg:privthr} shows the pseudocode of \privthr.
 
\privthr first introduces random noise to the count matrix $M$, similar to \privqt, and obtains a noisy count matrix $M'$ (Line 2).
\privthr then applies a wavelet transform on $M'$ to obtain $W'$ (Line 3).
$W'$ is then turned into a list $L'$ that keeps only positive values and the values in $L'$ is sorted in ascending order (Line 4).
Thus, only the positive values in $W'$ will be used for computing $k'$ based on the specified density threshold $p$.
To reduce the distortion of $k'$, starting from the smallest noisy positive values in $L'$, \privthr discards the first $\frac{|Z|'}{2}$ values (Line 6), where $Z$ represents the non-positive (negative or zero) values in the $W$ and $|Z|'$ is a noisy estimate of $|Z|$ (Line 5).
The reason why \privthr removes $\frac{|Z|'}{2}$ values from $L'$ is based on the utility analysis (in Section~\ref{subsec:utility}) that approximately $\frac{|Z|}{2}$ non-positive values in $W$ are turned into positive values due to the randomness of Laplacian noise.
Since $|Z|$ partially describes the data distribution and releasing $|Z|$ without protection may leak private information, \privthr also introduces Laplacian noise to $|Z|$, ensuring the whole process correctly enforces differentially privacy (Lines 11-17).
The noise introduced to $|Z|$ depends on the wavelet transform used to compute $W$.
For example, if we use Haar transform for $n$-dimensional data, a value in $W$ is computed by applying average for two neighboring elements along each dimension.
Since any single change in the input only causes one entry of the count matrix $M$ to change by 1,
the change of $M$ causes at maximum one value in $W$ to change, and thus causes $|Z|$ to change by 1 at maximum, i.e., the sensitivity of $|Z|$ is 1\footnote{For other wavelet transforms that use circular convolutions, such as Biorthogonal transform, the sensitivity of $n$ depends on the count of positive values and the count of negative values in the matrix computed by the coefficient vector~\cite{bior}.}.
Finally, \privthr obtains $d'$ as the top $k'$th value in $L''$ (Line 8), where any value in $L''$ greater than $d'$ is considered as a significant value,
and applies the connected component labeling algorithm to identify clusters of significant grids (Line 9).

\begin{algorithm}
\caption{\privthr}\label{alg:privthr}
\begin{algorithmic}[1]
\Input Dataset $D$, num\_grid $g$, density threshold $p$, wavelet transform $w$, differential privacy budget $\epsilon$, allocation percentage $\alpha$
\Output A set of differentially private clusters
\Procedure{$PrivTHR$}{$D,g,p,w,\epsilon,\alpha$}
   \State $M'$ = PrivQuantization($D,g,\alpha \epsilon$)
   \State $W'$ = WaveletTransform($M'$,$w$)
   \State $L'$ = ConvertToPosSortedArray($W'$)
   \State $|Z|'$ = \Call{NoisyCountOfNonPosValues}{$D,g,w,(1-\alpha)\epsilon$}
   \State $L''$ = RemoveFrom($L'$,0,$\frac{|Z|'}{2}$)
   \State $k'$ = $(1-p)|L''|$
   \State $d'$ = Top($k'$,$L''$)
   \State \textbf{return} ConnCompLabel($W',d'$)
\EndProcedure
\Procedure{NoisyCountOfNonPosValues}{$D, g, w,\epsilon$}\label{noisyzero}
\State $M$ = Quantization($D,g$)
\State $W$ = WaveletTransform($M$,$w$)
\State $|Z|$ = CountOfNonPos($W$)
\State $|Z|'$ = $|Z|$ + $Lap(\frac{Sensitivity(n)}{\epsilon})$
\State \textbf{return} $|Z|'$
\EndProcedure
\end{algorithmic}
\end{algorithm}

\textbf{Budget Allocation.}
\privthr first introduces Laplacian noise in the quantization step using a privacy budget $\epsilon_1 = \alpha \epsilon$, where $0 < \alpha < 1$.
In the significant grid identification step, \privthr further introduces Laplacian noise to $|Z|$ using the remaining privacy budget $(1-\alpha)\epsilon$.
Based on utility analysis in Section~\ref{subsubsec:privthr}, $\epsilon_1$ requires a smaller amount of budget than $\epsilon_2$.
Our empirical results in Section~\ref{section-experiments} further show in detail the impact of $\alpha$ on clustering accuracy. 

\subsection{Private Quantization with Noisy Threshold using Exponential Mechanism (\privthrem)}
Besides pruning noisy positive values in $W'$, we propose an alternative technique that employs Exponential mechanism for deriving $k'$ from the sorted list of $L$. 
Algorithm~\ref{alg:privthrem} shows the pseudocode of \privthrem.

\privthrem first introduces Laplacian noise to the count matrix $M$, which is similar to \privqt and \privthr. 
After that, we obtain a noisy count matrix $M'$ (Line 2) and the corresponding $W'$ (Line 3).
Different from the previous two techniques that compute $k'$ from $W'$, \privthrem derives $k'$ from $W$ using Exponential mechanism (Lines 7-15). 
In this case, although the sorted list derived from $W'$ is severely distorted in \privthrem, the derivation of $k'$ from $W$ is not affected by the distorted $W'$ at all.
Given reasonable privacy budget, $k'$ derived from $W$ using Exponential mechanism is reasonably accurate, compared to the case when $k'$ is derived from $W'$.

The quality function fed into the Exponential mechanism is~\cite{spatial}: 
$$
q(L,X)=-|rank(x)-k|,
$$
where $L$ represents the sorted positive values in $W$ with $Min$ and $Max$ values (Line 10), and $X$ represents the possible output space, i.e., all the possible values in the range of $(0, Max]$. 
Given a $W$ with $m$ positive values and their relationships are $x_1\geq x_2 \geq \ldots \geq x_m$,
these $m$ values divide the range $(0, Max]$ into $m$ partitions: 
$(0,x_m],(x_{m-1},x_{m-2}],$ $\ldots,(x_{2},x_1]$, and the ranks for these partitions are $m$, $m-1$, $\ldots$, 2, 1.
For any $x \in (x_{i-1}, x_{i}]$, its rank is $rank(x_{i})$. 
For example, if $x \in (x_{2},x_1]$, $rank(x) = rank(x_1) = 1$.
Similar to \privthr, when using Haar transform, any single change in the input causes only one value in $W$ to change. 
Thus, at maximum one value will be added into or removed from $L$, causing the outcome of $q(L, X)$ to be changed by 1, i.e., the sensitivity of $q(L, X)$ is 1\footnote{Similar to \privthr, for other wavelet transforms that use circular convolutions, the sensitivity of $q(L, X)$ depends on the count of positive values and the count of negative values in the matrix computed by the coefficient vector~\cite{bior}.}.

Plugging in the above quality function into Exponential mechanism, we obtain the following algorithm:
for any value $x\in (0, Max]$, the Exponential mechanism (EM) returns $x$ with probability \\ $Pr[EM(L) = x] \propto exp(-\frac{\epsilon |rank(x)-k|}{2})$  (Line 12). 
Since all the values in a partition have the same probability to be chosen, a random value from the partition $Pt_i = (x_{i-1}, x_{i}]$ will be chosen with the probability proportional to $|Pt_{i}|*exp(-\frac{\epsilon}{2}|i-k|)$. 
In other words, once $k'$ is chosen, \privthrem further computes a uniform random value $d'$ from $Pt_{i}$ (Line 13), and any value in $L'$ greater than $d'$ is considered as a significant value.

\textbf{Budget Allocation.}
Similar to \privthr, the privacy budget is split between two steps: introduction of Laplacian noise in quantization and obtaining $k'$ using Exponential mechanism. 
Previous empirical experiments~\cite{spatial} on splitting budgets between obtaining noisy median and noisy counts suggest that, 30\% vs. 70\% budget allocation strategy performs best. 
Specifically, 70\% of budget is allocated for obtaining noisy count matrix $M'$ (Line 2) and the remaining budget is allocated for computing $k'$ (Line 4).

\begin{algorithm}
\caption{\privthrem}\label{alg:privthrem}
\begin{algorithmic}[1]
\Input Dataset $D$, num\_grid $g$, density threshold $p$, wavelet transform $w$, differential privacy budget $\epsilon$, allocation percentage $\alpha$
\Output A set of differentially private clusters
\Procedure{$PrivTHR_{EM}$}{$D, g, p, w,\epsilon,\alpha$}
   \State $M'$ = PrivQuantization($D,g,\alpha \epsilon$)
   \State $W'$ = WaveletTransform($M'$,$w$)
   \State $d'$ = \Call{NoisyDensityThreshold}{$D, g, p, w,(1-\alpha) \epsilon$}
   \State \textbf{return} ConnCompLabel($W',d'$)
\EndProcedure

\Procedure{NoisyDensityThreshold}{$D, g, p, w,\epsilon$}\label{noisythresholdbyEM}
\State $M$ = Quantization($D,g$)
\State $W$ = WaveletTransform($M$,$w$)
\State $L$ = ConvertToPosSortedArray($W$)
\State $k$ = $(1-p)|L|$
\State $k'$ = ExponentialMechanism($L$,$k$,$\epsilon$)
\State $d'$ = UniformRandom($L$, $k'$ )
\State \textbf{return} $d'$
\EndProcedure
\end{algorithmic}
\end{algorithm}

\section{Privacy and Utility Analysis}\label{section-theoretical}
In this section, we present the theoretical analysis of proposed techniques \privqt, \privthr and \privthrem.

\subsection{Privacy Analysis}
In this part we establish the privacy guarantee of \privqt, \privthr and \privthrem.

\begin{theorem}
\privqt is $\epsilon$-differentially private.
\end{theorem}
\begin{proof}
\privqt introduces independent Laplacian noise $Lap(\frac{1}{\epsilon})$ to grid counts, which are computed on disjoint datasets. 
According to the parallel composition property of differential privacy described in Section~\ref{subsec:diffprivacy}, 
the privacy cost depends only on the worst guarantee of all computations over disjoint datasets. 
Therefore, \privqt is $\epsilon$-differentially private.
\end{proof}

\begin{theorem}
\privthr is $\epsilon$-differentially private.
\end{theorem}
\begin{proof}
\privthr splits privacy budget into two parts. First, for private quantization, adding Laplacian noise $Lap(\frac{1}{\alpha\epsilon})$ achieves strict $\alpha\epsilon$-differential privacy. The proof is same as \privqt. 
Second, \privthr introduces Laplacian noise $Lap(\frac{1}{(1-\alpha)\epsilon})$ to the true count of non-positive values in $W$, which achieves $(1-\alpha)\epsilon$-differential privacy. 
Using the composition property of differential privacy, \privthr achieves $\epsilon$-differentially private since $\epsilon = \alpha\epsilon + (1-\alpha)\epsilon$.
\end{proof}

\begin{theorem}
\privthrem is $\epsilon$-differentially private.
\end{theorem}
\begin{proof}
Similar to \privthr,  \privthrem has two steps of randomization: private quantization and obtaining noisy density threshold $d'$. 
Private quantization achieves $\alpha\epsilon$-differential privacy according to Laplace mechanism and parallel composition property. Sampling noisy density threshold $d'$ by Exponential mechanism consumes budget of $(1-\alpha)\epsilon$, which achieves $(1-\alpha)\epsilon$-differential privacy. According to the composition property of differential privacy, \privthrem is $\epsilon$-differentially private.
\end{proof}

\eat{
\begin{theorem}
\privqt is $\epsilon$-differentially private.
\end{theorem}

\begin{theorem}
\privthr is $\epsilon$-differentially private.
\end{theorem}

\begin{theorem}
\privthrem is $\epsilon$-differentially private.
\end{theorem}
}

\subsection{Utility Analysis}
\label{subsec:utility}
In this section, we present utility guarantees of our algorithms (\privqt, \privthr and \privthrem) with theoretical analysis.
In WaveCluster, the step of significant grid identification determines the clustering results. 
In the private results of WaveCluster, \privqt, \privthr and \privthrem return a list of noisy significant grids. 
To quantify the utility of \privqt, \privthr and \privthrem, we consider finding significant grids whose wavelet transformed values surpass a threshold to be similar to finding the top-$k$ frequent itemsets whose frequencies surpass a threshold. 
In significant grid identification, $L$ is the list of positive wavelet transformed values from $W$ sorted in ascending order, $Z$ represents the set of zero values from $W$, and $k$ indicates the threshold position in $L$ and all the top-$k$ values in $L$ correspond to significant grids, where $k = (1-p)|L|$.
One parameter to specify $k$ is the density threshold $p$, which remains the same either with or without noise introduction. However, $|L|$, another parameter to determine $k$, will be changed to $|L'|$ under differential privacy, where $L'$ is the list of positive wavelet transformed values from $W'$ sorted in ascending order. $L$ is different from $L'$ since noise introduction might result in a portion of zero values in $Z$ becoming positive and a small portion of positive values in $L$ becoming non-positive. 

\subsubsection{Utility Analysis for \privqt.}
\label{subsubsec:privthr}
We first provide the analysis of difference between $k$ and $k'$ in \privqt.
In \privqt, the difference between $k'$ and $k$ depends on two factors: 
(1) a set of zero values in $Z$ becoming noisy positive, $Z'_p = \{W'_{Z} | W'_{Z} = W_Z + Noise, W_Z \in Z, W'_{Z} > 0\}$, where $W'_{Z}$ is the noisy value of zero value in $Z$,
and (2) a set of positive values in $L$ becoming noisy non-positive, $L'_n = \{W'_{L} | W'_{L} = W_{L} + Noise, W_{L} \in L, W'_{L} \leq 0\}$, where $W'_{L}$ is the noisy value of positive value in $L$.
That is, $k' = (1-p) (|L| + |Z'_p| - |L'_n|)$.

\textbf{Analysis of $|Z'_p|$.}
In \privqt, since we are adding $Lap(\frac{1}{\epsilon})$ noise to each grid count and the Haar transform computes the average from four adjacent grids, the noise added into a wavelet transformed value is the sum of four i.i.d. samples from the Laplace distribution.
The sum of $h$ i.i.d. Laplace distributions with mean 0 is the difference of two i.i.d. Gamma distributions~\cite{kotz2001laplace},
referred to as distribution $T$. 
Distribution $T$ is a polynomial in $|x|$ divided by $e^{|x|}$,
which is a symmetric function and thus the probability for distribution $T$ to produce positive values is $\frac{1}{2}$. 
Thus, the events of values in $Z$ adding positive noise from distribution $T$ conform to the Binominal Distribution with parameters $|Z|$ and $\frac{1}{2}$ and its expected value is $\frac{|Z|}{2}$.

\textbf{Analysis of $|L'_n|$.}
For $L'_n$, each value is added the noise conforming to the symmetric distribution $T$.
The probability density function of $|L'_n|$ is
$f_{|L'_n|}(x) = {|L|\choose x} \prod_{i=1}^{x} f_T(y \leq -W_{i}) \prod_{i=x+1}^{|L|}\\
f_T(y > -W_{i}), W_{i} \in L$,
and its expected value $E[|L'_n|]$ is $\sum_{n=1}^{|L|} \\
f_T(y \leq -W_{i})$.
$E[|L'_n|]$ is large when $W_i$ is small and there is limited privacy budget.
Consider an extreme case that might not be suitable for clustering. 
All the positive values in $L$ are the minimum value 0.5 due to the sum of adjacent four grid counts being the minimum value 1,
resulting in a high $E[|L'_n|]$.
Clustering, especially WaveCluster algorithm, is useful when the dataset has dense areas (clusters) and empty areas (gap between clusters).
Such extreme case is not suitable for clustering since its data distribution is close to uniform distribution.
Those datasets that are interesting for clustering always have highly dense cluster centers and cluster borders with low density.
Only those values corresponding to border grids are possible to become noisy non-positive and the size of border grids is relatively small. 
Therefore, $E[|L'_n|]$ is a small constant.
We refer to the value of $|L'_n|$ as $\theta$ in the following analysis.

\textbf{Analysis of $k'-k$.}
In \privqt, $E[k'-k]=(1-p)(\frac{|Z|}{2} - \theta)$. There are two extreme cases when $k'-k \approx 0$.
For one extreme, $|Z| = |L|$ and all the positive values in $L$ is the minimum value 0.5.
When $\epsilon = 1$, $\theta \approx 0.43|L| \approx \frac{|Z|}{2}$, which makes $k'-k \approx 0$.
For another extreme, $|Z| = 0$, all the positive values in $L$ is large, e.g. $\geq 15$. 
When $\epsilon = 1$, $\theta \approx 0$ and $k'-k \approx 0$.
For those datasets that are interesting in the context of clustering, $|Z|$ is pretty large compared to the whole space since $Z$ is used to separate different clusters. 
What is more, dense areas within clusters are typically larger than the space of cluster borders with low density, i.e. $\theta$ is far smaller than $\frac{|Z|}{2}$.
In \privqt, $\frac{|Z|}{2}$ dominates the difference between $k'-k$, which increases false positive rate.
In \privthr and \privthrem, we use different strategies to minimize the difference between $k'-k$.

\eat{
\begin{lemma}
In \privqt, the expected value of $k'$ is $(1-p)(|L| + \frac{|Z|}{2} - \theta)$.
\end{lemma}
}

\begin{theorem}
In \privqt with Haar transform, given $0<\omega < 1$, let $\eta_1 = \frac{|Z|}{2} - \sqrt{\frac{|Z|\ln{(\frac{1}{\omega})}}{2}}$, $\eta_2=\frac{|Z|}{2} + \sqrt{\frac{|Z|\ln{(\frac{1}{\omega})}}{2}}$, and $\gamma = \frac{8}{\epsilon}\ln{(\frac{4(|L|+|Z|)}{\omega})}$, then with probability at least $(1-\omega)^2$,
(1) all values in $L$ greater than $W_{k_{min}'} + \gamma$ are output, where $k_{min}'=k+(1-p)(\eta_1-\theta)$, and
(2) no values in $L$ less than $W_{k_{max}'} - \gamma$ are output, where $k_{max}'=k+(1-p)(\eta_2-\theta)$.
\end{theorem}

\begin{proof}
In \privqt, $k' = (1-p)(|L| + |Z'_p| - |L'_n|)$.
Since $|Z'_p|$ follows Binominal distribution with parameters $|Z|$ and $\frac{1}{2}$ and $|L'_n|$ is noted as a small value $\theta$, $k'$ follows the Binomial distribution and decides the number of values in $L$ that become output.
Given $\omega$, we can derive $k'$'s lower bound $k'_{min}$,
and show that values greater than $W_{k_{min}'} + \gamma$ are output, i.e., subclaim (1). 
Let $1-\omega = Pr(k' \geq k_{min}') = Pr(|Z'_p| \geq \eta_1)$.
As $Pr(|Z'_p| \geq \eta_1)$ = 1 - $Pr(|Z'_p| \leq \eta_1)$
and $Pr(|Z'_p| \leq \eta_1) \leq e^{(-2\frac{(\frac{|Z|}{2} - \eta_1)^2}{|Z|})}$~\cite{probabilisticmethod}, 
we have $\eta_1  = \frac{|Z|}{2} - \sqrt{\frac{|Z|\ln{(\frac{1}{\omega})}}{2}}$.
For constant $\omega$, $\eta_1 = O(\frac{|Z|}{2} - \sqrt{\frac{|Z|}{2}})$ will suffice.

Similar as $k'_{min}$,
we can also derive the bound $\gamma$ of the noise added to each value in $L \cup Z$ based on $\omega$.
For Haar wavelet transform, each value in $L \cup Z$ is added the noise 
that is the sum of 4 Laplacian random variables divided by 2 (i.e., $\frac{4Lap(\frac{1}{\epsilon})}{2}$).
For values in $L \cup Z$,
let all $4(|L|+|Z|)$ Laplacian random variables generate noise within $[-\frac{\gamma}{4},\frac{\gamma}{4}]$.
The probability that no Laplacian random variable' value is outside $[-\frac{\gamma}{4},\frac{\gamma}{4}]$ is $1- Pr(A)$, 
where $A$ is that at least one Laplacian random variable's value is outside $[-\frac{\gamma}{4},\frac{\gamma}{4}]$.
By union bound, $Pr(A) \leq \sum_{i=1}^{4(|L|+|Z|)} Pr(B_i)$, where $B_i$ is that $i$th Laplacian random variable's noise is outside $[-\frac{\gamma}{4},\frac{\gamma}{4}]$ and $Pr(B_i)= e^{-\frac{\epsilon \gamma}{8}}$.
Thus, we can derive that with at least the probability $1-4(|L|+|Z|)e^{-\frac{\epsilon \gamma}{8}}$,
no Laplacian random variable' value is outside $[-\frac{\gamma}{4},\frac{\gamma}{4}]$,
and each value in $L \cup Z$ has their noise amount within $[-\frac{\gamma}{2},\frac{\gamma}{2}]$.
Let $\omega = 4(|L|+|Z|)e^{-\frac{\epsilon \gamma}{8}}$, 
then $-\frac{\epsilon \gamma}{8} = \ln{(\frac{\omega}{4(|L|+|Z|)})}$
and we have $\gamma = \frac{8}{\epsilon}\ln{(\frac{4(|L|+|Z|)}{\omega})}$.
For constant $\omega$, $\gamma = O(\frac{ln{(4(|L|+|Z|))}}{\epsilon})$.

Subclaim (1) can be derived based on 
(a) with probability at least $1-\omega$, $k' \geq k_{min}'$
and (b) with probability at least $1-\omega$, 
the noise of each value in $L$ being within $[-\frac{\gamma}{2}, \frac{\gamma}{2}]$.
Detailed proof is omitted here.
Subclaim (1) requires both conditions (a) and (b) to hold,
and thus the probability is at least $(1-\omega)^2$.

We can derive the upper bound $k_{max}'$ of $k'$ given $\omega$.
Let $1-\omega = Pr(k' \leq k_{max}') = Pr(|Z'_p| \leq \eta_2)$.
Recall that $|Z'_p|$ follows Binomial distribution $(|Z|, \frac{1}{2})$,
and Binomial distribution $(|Z|, \frac{1}{2})$ is symmetric with respect to $\frac{|Z|}{2}$. 
Thus, the probability of sampling a value from the range $[0, \eta_2]$ is the same as 
sampling a value from the range $[\eta_1, |Z|]$,
and we have $\eta_2 = |Z|-\eta_1=\frac{|Z|}{2} + \sqrt{\frac{|Z|\ln{(\frac{1}{\omega})}}{2}}$.
For constant $\omega$, $\eta_2 = O(\frac{|Z|}{2} + \sqrt{\frac{|Z|}{2}})$ will suffice.

Subclaim (2) can be proved based on 
(c) with probability at least $1-\omega$, $k' \leq k_{max}'$
and (b) with probability at least $1-\omega$, 
the noise of each value in $L$ being within $[-\frac{\gamma}{2}, \frac{\gamma}{2}]$.
As subclaim (2) requires both conditions (c) and (b) to hold,
the probability is at least $(1-\omega)^2$. 
\end{proof}
For other wavelet transforms that use circular convolutions, such as Biorthogonal transform, the derivation for the bounds of $k'$ with $\eta_1$ and $\eta_2$ remains the same since $|Z'_p|$ following Binomial distribution is independent of any wavelet transform being adapted.
Thus, our framework is extensible to other wavelet transforms, and the bound of noise magnitude $\gamma$ depends on the amount of adjacent grid counts involved in computing a wavelet transformed value.

\subsubsection{Utility Analysis for \privthr.}
\label{subsubsec:privthr}
\eat{
\begin{lemma}
In \privthr, the expected value of $k'$ is $(1-p)(|L| - \theta)$, which is close to $k$.
\end{lemma}

\textsc{Theorem} 5. In \privthr, the expected value of false positive rate $FPR$ is close to 0.

Based on \textsc{Theorem} 4, the expected value of $FPR$ depends on the expected value of the difference between $k'$ and $k$. 
From LEMMA 2, we can deduce that the expected value of $FPR$ for \privthr is $FPR \approx 0$.
}

\begin{theorem}
In \privthr with Haar transform, given $0<\omega < 1$, let $\eta_1 = \frac{|Z|}{2} - \sqrt{\frac{|Z|\ln{(\frac{1}{\omega})}}{2}}$, $\eta_2=\frac{|Z|}{2} + \sqrt{\frac{|Z|\ln{(\frac{1}{\omega})}}{2}}$, $\gamma = \frac{8}{\epsilon_1}\ln{(\frac{4(|L|+|Z|)}{\omega})}$ and $\beta = \frac{2}{\epsilon_2}\ln{(\frac{1}{\omega})}$, then with probability at least $(1-\omega)^3$,
(1) all values in $L$ greater than $W_{k_{min}'} + \gamma$ are output, where $k'_{min} = k+(1-p)(\eta_1-\theta-\frac{|Z|}{2}-\beta)$, and
(2) no values in $L$ less than $W_{k_{max}'} - \gamma$ are output, where $k'_{max} = k+(1-p)(\eta_2-\theta-\frac{|Z|}{2}+\beta)$.
\end{theorem}

\begin{proof}
In \privthr, we allocate $\epsilon_1$ for private quantization and $\epsilon_2$ for protecting $\frac{|Z|}{2}$, 
which makes $k' = (1-p) (|L| + |Z'_p| -\theta - \frac{|Z|}{2} + Lap(\frac{1}{\epsilon_2}))$.
With the probability at least $1-e^{-\frac{\epsilon_2 \beta}{2}}$,
$Lap(\frac{1}{\epsilon_2})$ has the noise amount within $\beta$.
Let $1-\omega =1-e^{-\frac{\epsilon_2 \beta}{2}}$, then we get $\beta = \frac{2}{\epsilon_2}\ln{(\frac{1}{\omega})}$.
For constant $\omega$, $\beta = O(\frac{1}{\epsilon_2})$ will suffice.
The proofs of $\eta_1$, $\eta_2$, $\gamma$ and subclaims (1) and (2) are the same as \textsc{Theorem} 4,
and $\gamma = O(\frac{ln{(4(|L|+|Z|))}}{\epsilon_1})$ for constant $\omega$.
\end{proof}

\textbf{Difference between \privthr and \privqt:}
By \textsc{Theorem} 4, in \privqt $k_{min}'=k+(1-p)(\eta_1-\theta)$ and $k_{max}'=k+(1-p)(\eta_2-\theta)$.
By \textsc{Theorem} 5, in \privthr $k_{min}'= k+(1-p)(\eta_1-\theta-\frac{|Z|}{2}-\beta)$ and $k_{max}'= k+(1-p)(\eta_2-\theta-\frac{|Z|}{2}+\beta)$.
$\eta_1 = O(\frac{|Z|}{2} - \sqrt{\frac{|Z|}{2}})$, $\eta_2 = O(\frac{|Z|}{2} + \sqrt{\frac{|Z|}{2}})$.
As we can see, by removing $\frac{|Z|}{2} \pm \beta$ positive values from $L'$, 
\privthr provides better utility guarantee than \privqt since the difference between $k'$ and $k$ becomes $O(\sqrt{\frac{|Z|}{2}}) \pm (\beta + \theta)$, where $\theta$ is a small constant and $\beta$ is small when sufficient budget $\epsilon_2$ is provided.

\subsubsection{Utility Analysis for \privthrem.}
\label{subsubsec:privthrem}

\begin{theorem}
In \privthrem with Haar transform, given $0<\omega < 1$, 
let $\eta_1 = |L|-k-1+\frac{2}{\epsilon_2}\ln(\frac{|W_{max}-W_k|}{|W_{k}|\omega})$,
$\eta_2 = k-2+\frac{2}{\epsilon_2}\ln(\frac{|W_{k}|}{|W_{max}-W_k|\omega})$, 
and $\gamma = \frac{8}{\epsilon_1}\ln{(\frac{4(|L|+|Z|)}{\omega})}$, 
then with probability at least $(1-\omega)^2$,
(1) all values in $L$ greater than $W_{k_{min}'} + \gamma$ are output, where $k_{min}'=k-\eta_1$, and
(2) no values in $L$ less than $W_{k_{max}'} - \gamma$ are output, where $k_{max}'=k+\eta_2$.
\end{theorem}

\begin{proof}
In \privthrem, we allocate $\epsilon_2$ for deriving $k'$ from $k$ by employing Exponential mechanism, a general method proposed in~\cite{exponential}. 
The probability of selecting a rank $i$ is $|Pt_{i}|*exp(-\frac{\epsilon_2}{2}|i-rank(W_k)|)$, where $Pt_i$ is the range $(W_{i-1}, W_{i}]$ decided by the $i-1$th and $i$th wavelet transformed values.

Let $1-\omega$ be the probability of sampling a $k'$ where $k - k' \leq \eta_1$, then
\vspace*{-2ex}
$$\Leftrightarrow \omega < \frac{|W_{max} - W_{k}| * e^{-\frac{\epsilon_2}{2}(\eta_1+1)}}{|W_k|*e^{-\frac{\epsilon_2}{2}(|L|-k)}} $$
\vspace*{-2ex}
$$\Leftrightarrow \eta_1 < |L|-k-1+\frac{2}{\epsilon_2}\ln(\frac{|W_{max}-W_k|}{|W_{k}|\omega})$$
For constant $\omega$, $\eta_1 = O(|L|-k+\frac{1}{\epsilon_2}\ln(\frac{|W_{max}-W_k|}{|W_{k}|})$ will suffice. 

Let $1-\omega$ be the probability of sampling a $k'$ where $k' - k \leq \eta_2$, then
\vspace*{-2ex}
$$\Leftrightarrow \omega < \frac{|W_{k}| * e^{-\frac{\epsilon_2}{2}(\eta_2+1)}}{|W_{max}-W_k|*e^{-\frac{\epsilon_2}{2}(k-1)}} $$
\vspace*{-2ex}
$$\Leftrightarrow \eta_2 < k-2+\frac{2}{\epsilon_2}\ln(\frac{|W_{k}|}{|W_{max}-W_k|\omega})$$
For constant $\omega$, $\eta_2 = O(k+\frac{1}{\epsilon_2}\ln(\frac{|W_{k}|}{|W_{max}-W_k|})$ will suffice.
The proof of $\gamma$ and subclaims (1) and (2) are the same as \textsc{Theorem} 4.
\end{proof}

\textbf{Analysis of \privthr and \privthrem}.
By \textsc{Theorem} 5 and \textsc{Theorem} 6, 
the accuracy for sampling $k'$ in \privthr is dominated by $\frac{1}{\epsilon_2}$
while in \privthrem the accuracy is dominated by $\frac{1}{\epsilon_2}\ln(\frac{|W_{k}|}{|W_{max}-W_k|})$.
Depending on the data distribution, \privthrem may present better or worse utility guarantee than \privthr:  \\
$\ln(\frac{|W_{k}|}{|W_{max}-W_k|})$ is positive when when $\frac{|W_{k}|}{|W_{max}-W_k|} > 1$,
and the accuracy for sampling $k'$ in \privthrem becomes more sensitive to $\epsilon_2$ than \privthr;
$\ln(\frac{|W_{k}|}{|W_{max}-W_k|})$ becomes negative when $\frac{|W_{k}|}{|W_{max}-W_k|}$ is less than 1, 
and the bounds of utility guarantee for \privthrem becomes better than \privthr.

Section~\ref{section-experiments} demonstrates that by reducing the difference between $k'$ and $k$, \privthr and \privthrem achieve more accurate results than \privqt, which conforms to the above analysis.

\begin{figure*}
\centering
\subfloat[\textit{DS1}, \textit{g} = 64, \textit{p} = 58]{\label{DS1}\includegraphics[scale=0.34,clip]{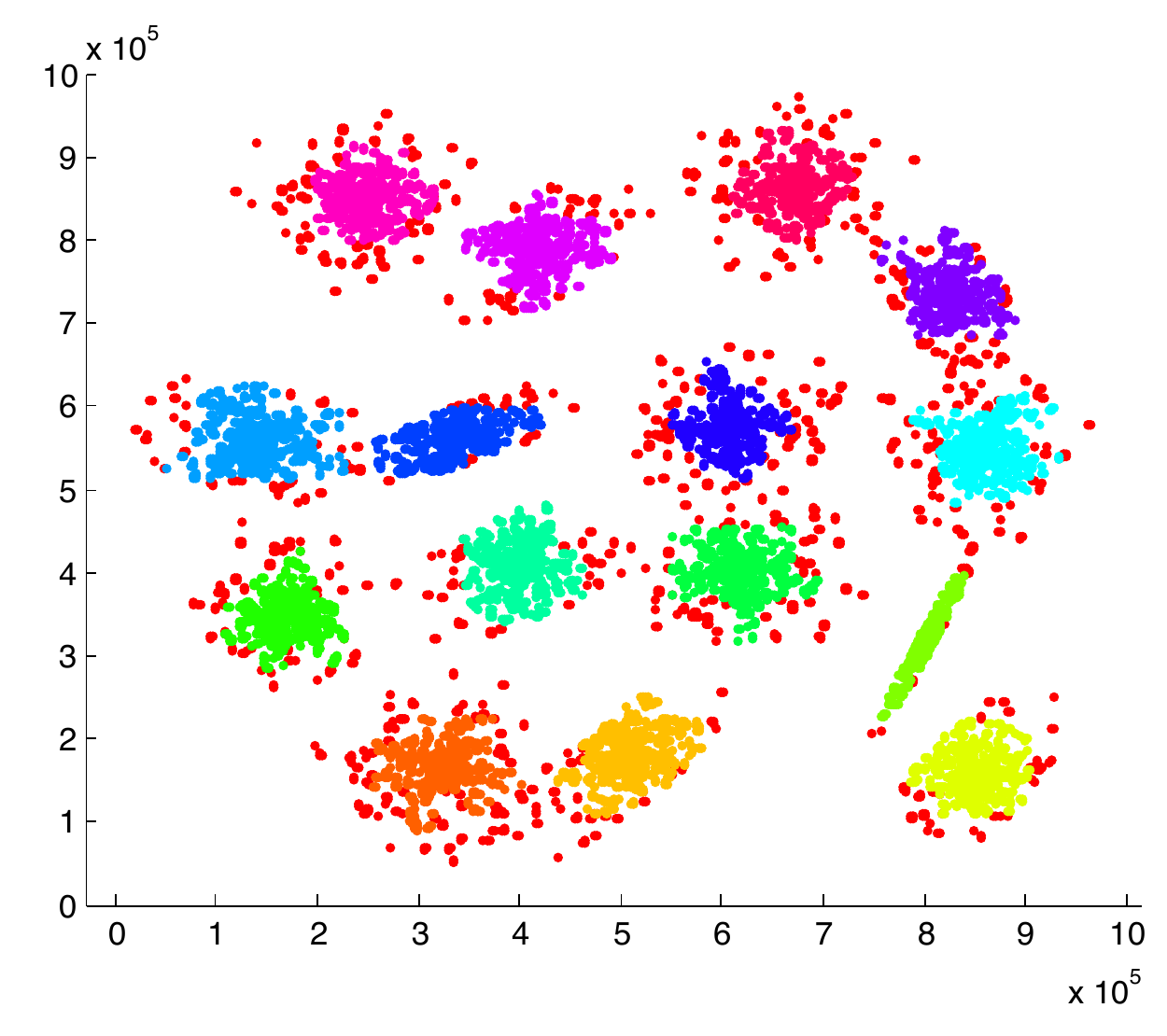} }
\subfloat[\textit{DS2}, \textit{g} = 40, \textit{p} = 10]{\label{DS2}\includegraphics[scale=0.34,clip]{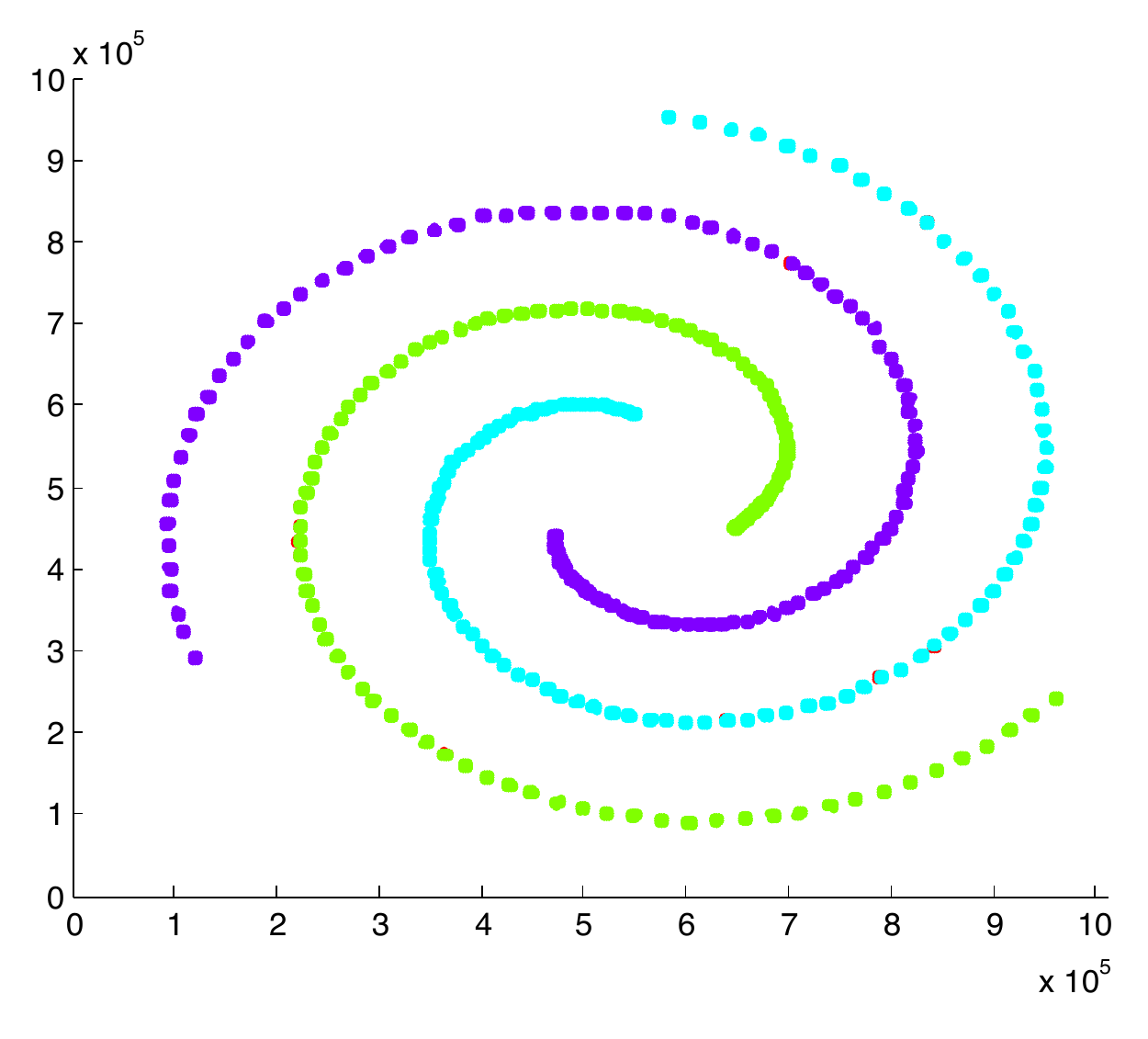} }
\subfloat[\textit{DS3}, \textit{g} = 36, \textit{p} = 23]{\label{DS3}\includegraphics[scale=0.34,clip]{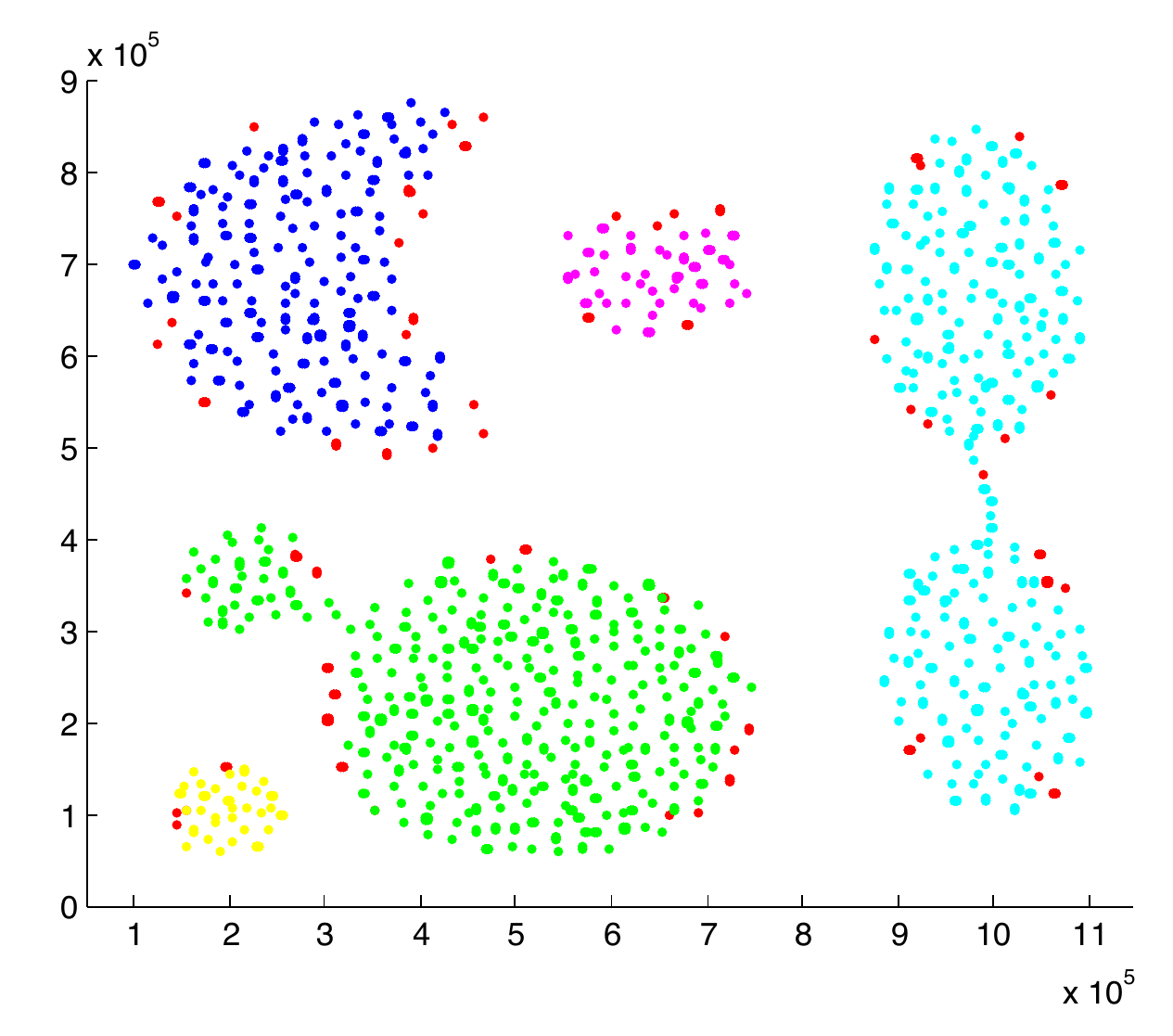} }
\subfloat[\textit{Gowalla}, \textit{g} = 80, \textit{p} = 31]{\label{DS3}\includegraphics[scale=0.34,clip]{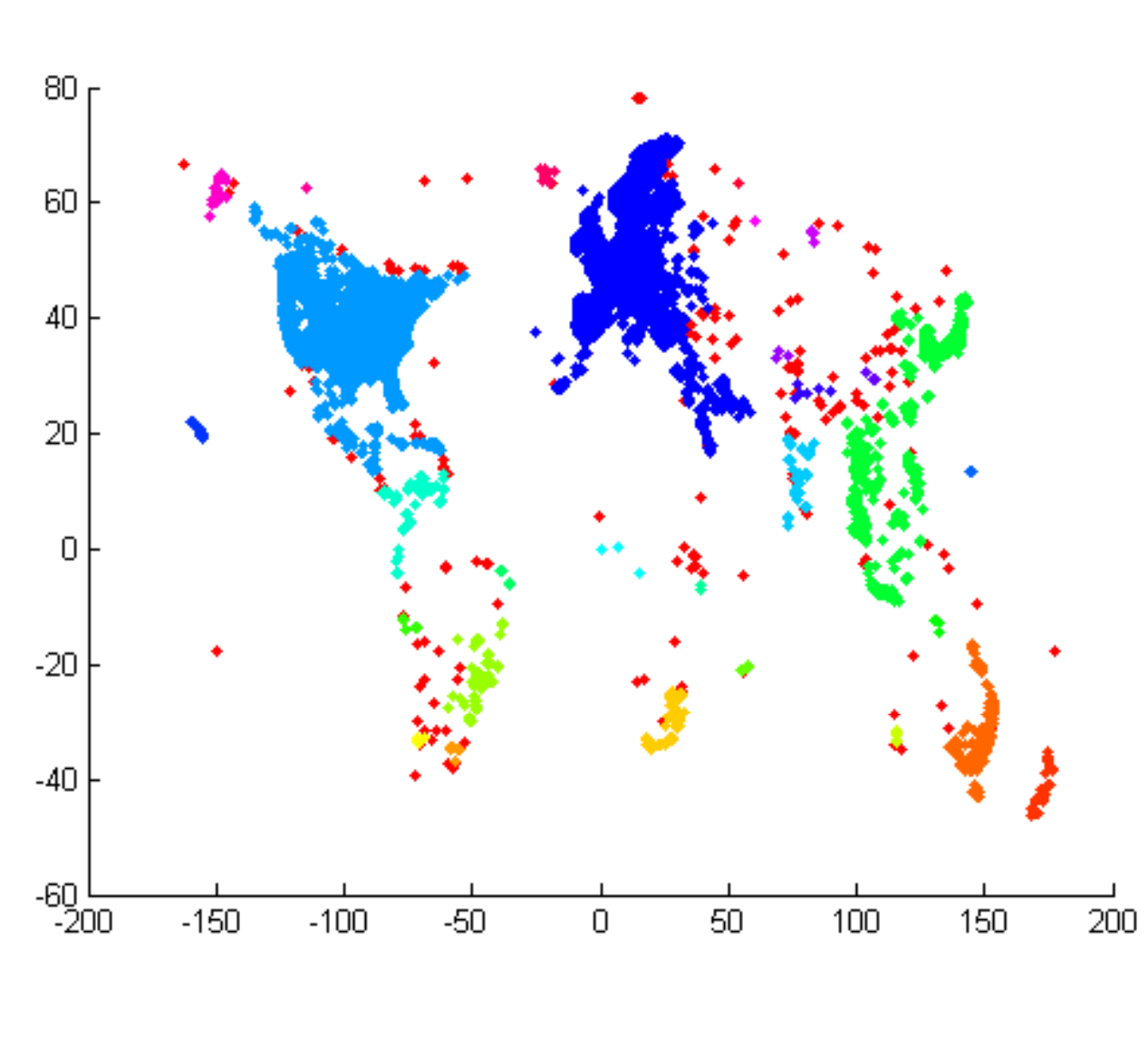} } \\
\vspace*{-1ex}
\caption{\label{fig:3datasets}Illustration of datasets and their WaveCluster results.} 
\end{figure*}

\section{Quantitative Measures}\label{section-measurements}
To quantitatively assess the utility of differentially private WaveCluster, 
we propose two types of measures for measuring the dissimilarity between true and differentially private WaveCluster results. 
The first type, $DSG_C$, measures the dissimilarity of the significant grids and the clusters between true and private results.
The second type focuses on observing the usefulness of differentially private WaveCluster results for further data analysis.
The reason is that a slight difference in the significant grids or clusters may cause a significant difference when using the WaveCluster results.
In this paper, we choose a typical application of further data analysis: building a classifier from the clustering results to predict unlabeled data~\cite{multiscaling}.
The classifier built from true WaveCluster results is called the true classifier $clf_t$ while the classifier built from differentially private WaveCluster results is called the private classifier $clf_p$.
To measure the dissimilarity between $clf_t$ and $clf_p$, we propose two metrics: $OCM$ and $2CE$.

\subsection{Dissimilarity based on Significant Grids and Clusters}

$DSG_C$ considers the dissimilarities of significant grids and clusters.
Assume that there are $t$ clusters of true significant grids and $s$ clusters of differentially private significant grids. 
$t$ might not be equal to $s$, and the cluster labels in $t$ true clusters and $s$ private clusters are completely arbitrary. 
To accommodate these differences, we adopt the Hungarian method~\cite{Kuhn1956}, a combinatorial optimization algorithm, to solve the matching problem between $t$ true clusters and $s$ private clusters while minimizing the matching difference. 

When cluster $C_i$ matches to cluster $C_j$, we define that the distance $d$ between cluster $C_i$ and cluster $C_j$ is $max\{|C_i\backslash C_j|, |C_j\backslash C_i|\}$.
Consider a cluster $C_i = \{g_{1}, g_3, g_5\}$ and a cluster $C_j$ = $\{g_1, g_5, g_7, g_9\}$.
The distance $d$ between clusters $C_i$ and $C_j$ is $max\{|\{g_3\}|, |\{g_7,g_9\}|\}$
$ = 2$.
Given $t$ true clusters, $s$ private clusters, and $t \geq s$ , 
a matching $M_{t,s}$ of $t$ true clusters and $s$ private clusters is a set of cluster pairs, where each private cluster is matched with a true cluster.
We then define the cost of a matching ($M_{cost}$) as the sum of all the distances between each cluster pair in the matching $M_{t,s}$ plus the count of significant grids in the non-matched clusters:

\vspace*{-2ex}
$$M_{cost} = \sum_{1\leq i_x \leq t, 1\leq j_y \leq s} max\{|C_{i_x}\backslash C_{j_y}|, |C_{j_y}\backslash C_{i_x}|\} + \sum_{1\leq z \leq t} |C_z|$$
\vspace*{-2ex}

Here, $i_x$ and $j_y$ indicate the subscripts of clusters in a matched pair.
$|C_z|$ represents the count of significant grids in the non-matched true clusters.
Among all the possible matchings of clusters, we use the Hungarian method to find the optimal matching with the minimum $M_{cost}$,
and computed $DSG_C$ as: 

\vspace*{-3ex}
$$DSG_C = \frac{M_{cost}}{|T|}$$
\vspace*{-2ex}

Here $T$ denotes the set of significant grids in the true WaveCluster results.

\subsection{Dissimilarity based on Classifier Prediction}
\label{subsec:OCM2CE}
$OCM$ and $2CE$ measure the dissimilarity between $clf_t$ and $clf_p$.
We name this way of evaluation as ``clustering-first-then-classification'':
given a set of unlabeled data points, we use a portion of the data points (e.g., 90\%) to compute WaveCluster results, where each cluster is a set of significant grids.
Using the significant grids with cluster labels as training data, we build classifiers $clf_t$ and $clf_p$,
and use them to predict the classes for the remaining data points (e.g., 10\%).

\textbf{Dissimilarity of Classifiers based on Optimal Class Matching ($OCM$).}
$OCM$ measures the dissimilarity between the two sets of classes predicted by $clf_t$ and $clf_p$ for the same test samples.
We use $L_{t}$ to denote the set of classes predicted by $clf_t$ and $L_{p}$ to denote the set of classes predicted by $clf_p$.
Since $L_{t}$ and $L_{p}$ are completely arbitrary, we exploit the Hungarian method to find the optimal matching between $L_{t}$ and $L_{p}$.

Assume that a class $L_{t,i}$ predicted by $clf_t$ is matched to a class  $L_{p, j}$ predicted by $clf_p$, forming a class pair.
We compute the count of common test samples in the class $L_{t,i}$ and the class $L_{p,j}$, and sum the common test samples in each class pair to compute $CT$: 
$$CT =  \sum_{1\leq i \leq c_1, 1\leq j \leq c_2} |L_{t,i}\cap L_{p,j}|  $$

Here $c_1$ is the count of classes in $L_t$ and $c_2$ is the count of classes in $L_p$, and we assume $c_1 \geq c_2$. 
Since there are many possible mappings from the classes in $L_t$ to the classes in $L_p$, we use the Hungarian method to find the optimal mapping that maximizes $CT$.
Based on $CT$ and the total count of the test samples $TT$, we  derive the dissimilarity $OCM$:
\vspace*{-2ex}
$$OCM = 1-\frac{CT}{TT}$$
\vspace*{-2ex}

When the dissimilarity is smaller, the differentially private WaveCluster results are more similar to the true WaveCluster results and maintain high utility for classification use.

\textbf{Dissimilarity of Classifiers based on 2-Combination Enumeration ($2CE$).}
$2CE$ measures the dissimilarity between $clf_t$ and $clf_p$ based on relationships of every pair of test samples, i.e., whether two samples are in the same class.
Essentially, given a pair of test samples $A$ and $B$, we say $A$ and $B$ are classified consistently either (1) $clf_t(A)=clf_t(B)$ and $clf_p(A)=clf_p(B)$ or (2) $clf_t(A) \neq clf_t(B)$ and $clf_p(A) \neq clf_p(B)$. 
$2CE$ is the ratio of the count of test sample pairs that are not classified consistently over the total number of test sample pairs, which is the set of 2-combination of the test samples.
$2CE$ uses pairs of test samples to eliminate the need of finding the optimal matching between the classes predicted by $clf_t$ and $clf_p$.

\begin{figure*}
\vspace*{1ex}
\centering
\subfloat[\textit{k' and k on DS1}]{\includegraphics[scale=0.48,clip]{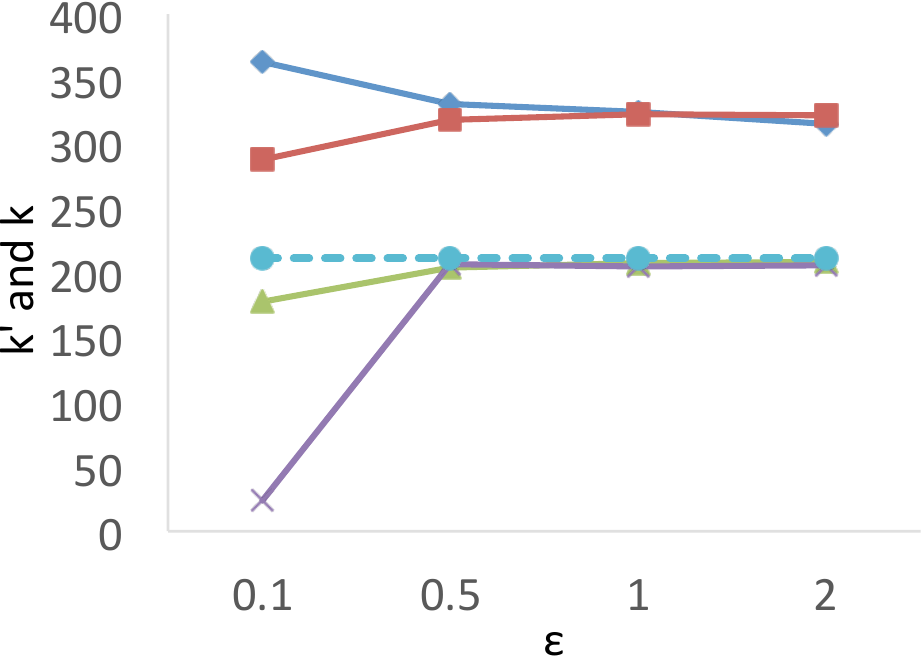} }
\subfloat[\textit{k' and k on DS2}]{\includegraphics[scale=0.48,clip]{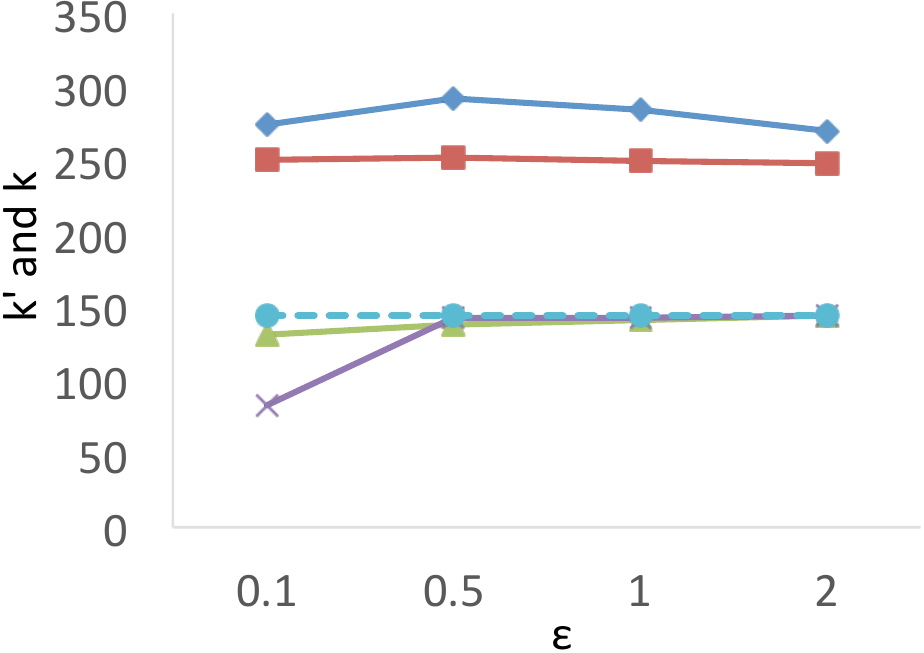} }
\subfloat[\textit{k' and k on DS3}]{\includegraphics[scale=0.48,clip]{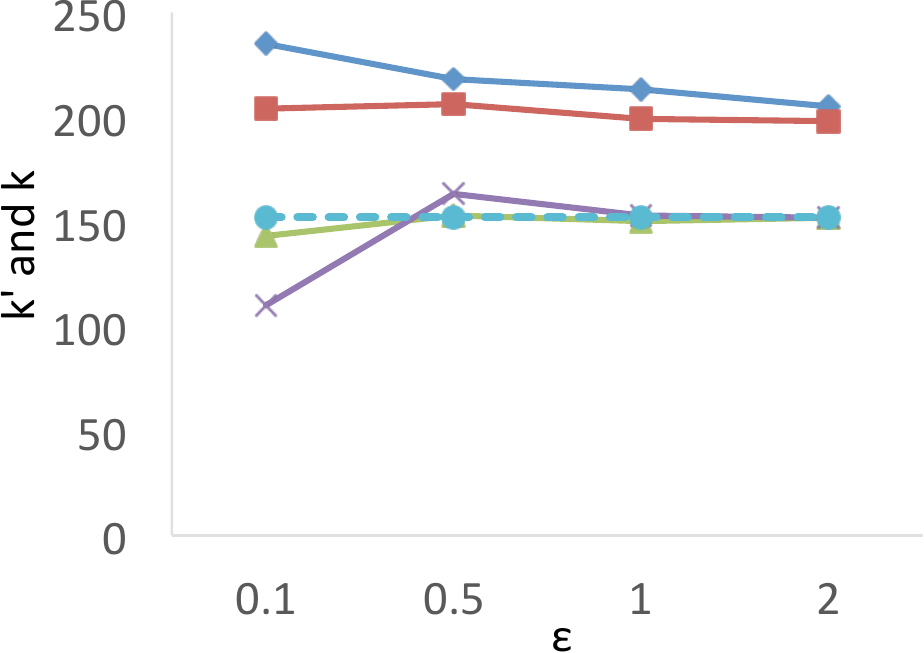} }
\subfloat[\textit{k' and k on Gowalla}]{\includegraphics[scale=0.48,clip]{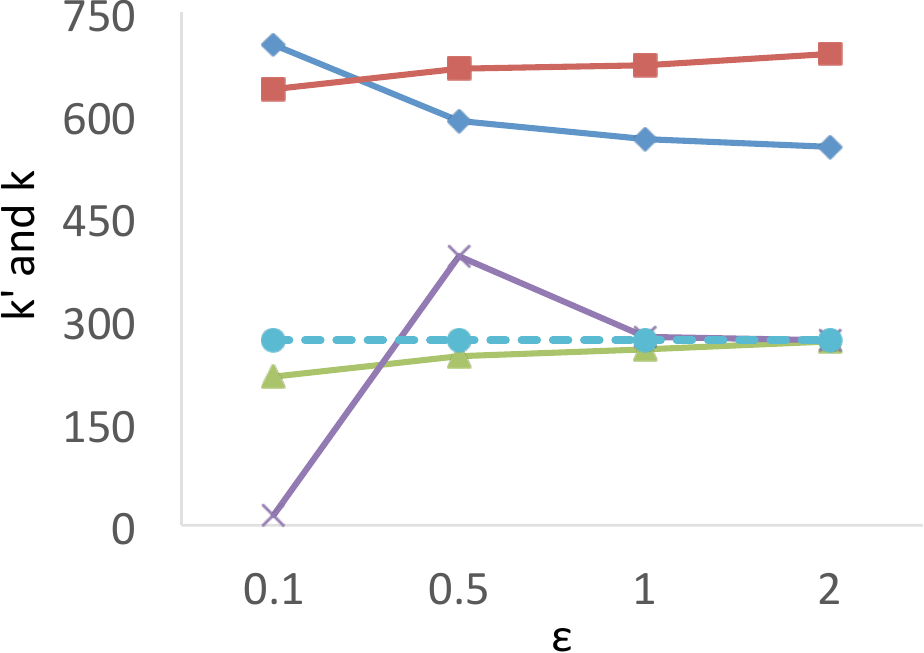} }\\
\subfloat{\includegraphics[scale=0.8,clip]{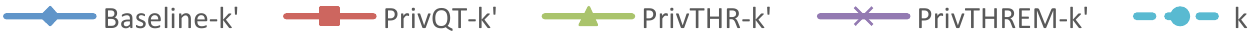} } \\
\vspace*{-1ex}
\caption{\label{fig:comparisonofk}Comparing private $k'$ of 4 techniques with true $k$ on $DS1$, $DS2$, $DS3$ and $Gowalla$ with increasing $\epsilon$.} 
\vspace*{-4ex}
\end{figure*}

\section{Experiments}\label{section-experiments}
We evaluate the proposed techniques using three datasets that are widely used in previous clustering algorithms~\cite{datasetlink}, and one large scale dataset derived from the check-in information in Gowalla\footnote{https://snap.stanford.edu/data/loc-gowalla.html.} geo-social networking website~\cite{Gowalla}, which was used to evaluate grid-based clustering algorithms in~\cite{dbscansocial}.

\subsection{Experiment Setup}
In our experiments, we compare the performances of the four techniques, \baseline, \privqt, \privthr, and \privthrem, on the four datasets using two types of measures proposed in Section~\ref{section-measurements} and provide analysis on the results. 
We use Haar transform as the wavelet transform and set the wavelet decomposition level to 1 for the four techniques.
\baseline uses the adaptive-grid method~\cite{geospatial} for synthetic data generation.
The classification algorithm used for measuring $OCM$ and $2CE$ is C4.5 decision tree algorithm~\cite{c4.5}.
We conduct experiments with privacy budgets ranging from 0.1 to 2.0;
for each budget and each metric, we apply the techniques on each dataset for 10 times and compute their average performances.
All experiments were conducted on a machine with Intel 2.67GHz CPU and 8GB RAM.

\textbf{Datasets.}
The four clustering datasets contain different data shapes that are specially interesting for clustering.
Figures~\ref{fig:3datasets} shows the WaveCluster results on four datasets under certain parameter settings of grid size $g$ and density threshold $p$.
Any two adjacent clusters are marked with different colors. 
The points in red color are identified as noise, which fall into the non-significant grids.

$DS1$ is a dataset containing 15 Gaussian clusters with different degrees of cluster overlapping.
It contains 30000 data points.
These 15 clusters are all in convex shapes. 
The center area of each cluster has higher density and is resistant to noise. 
However, the overlapped area of two adjacent clusters has lower density and is prone to be affected by noise, which might turn the corresponding non-significant grids into significant grids and further connect two separate clusters.
$DS2$ is a dataset with 3 spiral clusters.
It contains 31200 data points.
The head of each spiral is quite close to one another. 
Some noisy significant grids are very likely to bridge the gap between adjacent spirals and merge them into one cluster.
$DS3$ is a data dataset with 5 various shapes of clusters, including concave shapes.
It contains 31520 data points.
There are two clusters that both contain two sub components and a narrow line-shape area that bridges those two sub components. 
The narrow bridging area has low density and might be turned into non-significant grids, causing a cluster to split into two clusters.
$Gowalla$ is the check-in dataset resembling the world map, which records time and location information of users' check-ins. 
We use only the location information for evaluation.
There are about 6.4M records in total.
The large size of the dataset makes it infeasible to run experiments with C4.5 and \baseline due to memory constraints.
Thus, similar to~\cite{geospatial}, we sampled 1M records from the dataset for evaluation.

We next present the results on comparing $k'$ and $k$,
and then present the results of the two types of measures.

\begin{figure*}
\vspace*{-2ex}
\centering
\subfloat[\textit{DSG$_C$-DS1}]{\includegraphics[scale=0.48,clip]{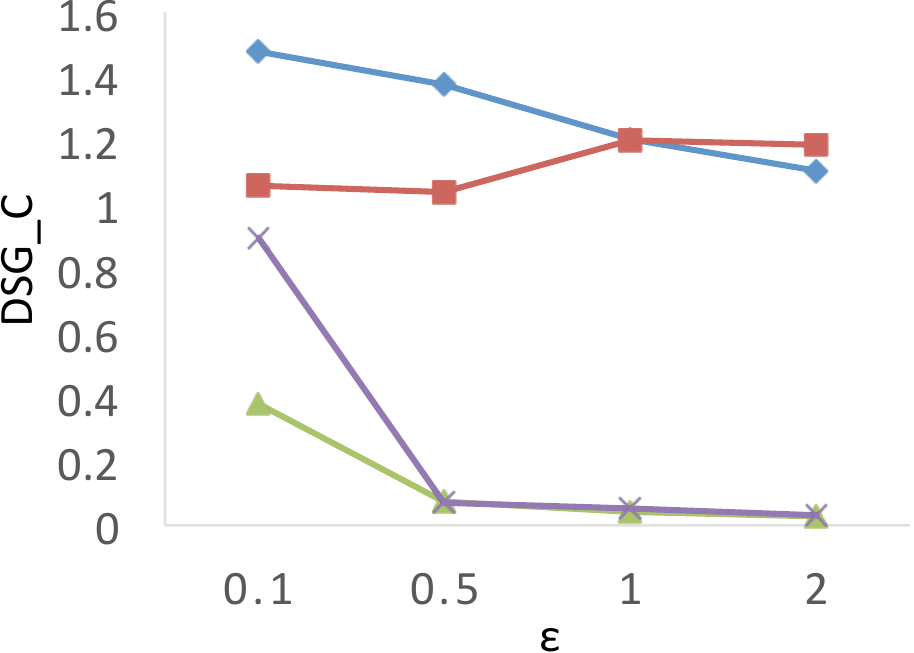} }
\subfloat[\textit{DSG$_C$-DS2}]{\includegraphics[scale=0.48,clip]{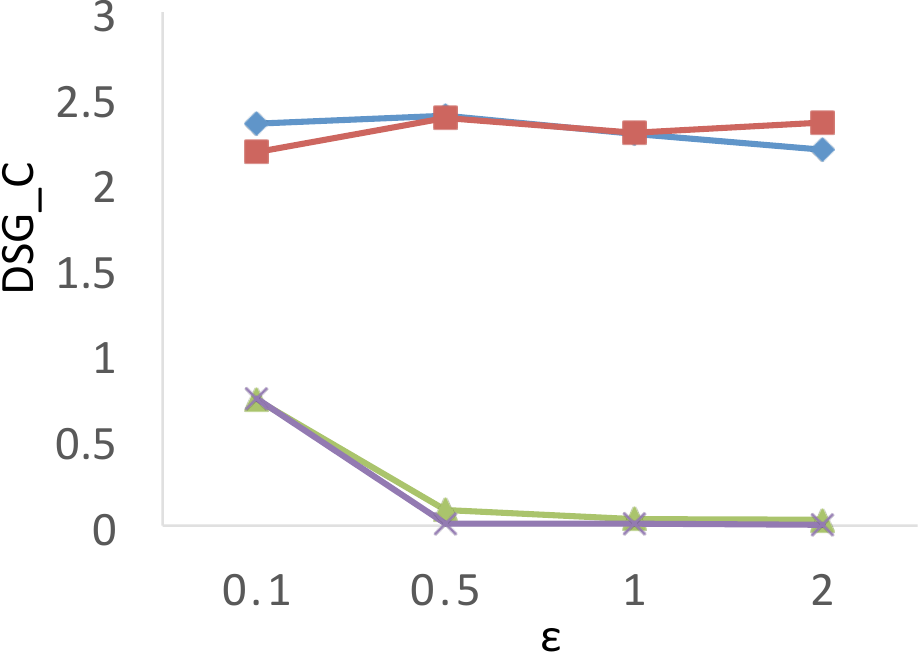} }
\subfloat[\textit{DSG$_C$-DS3}]{\includegraphics[scale=0.48,clip]{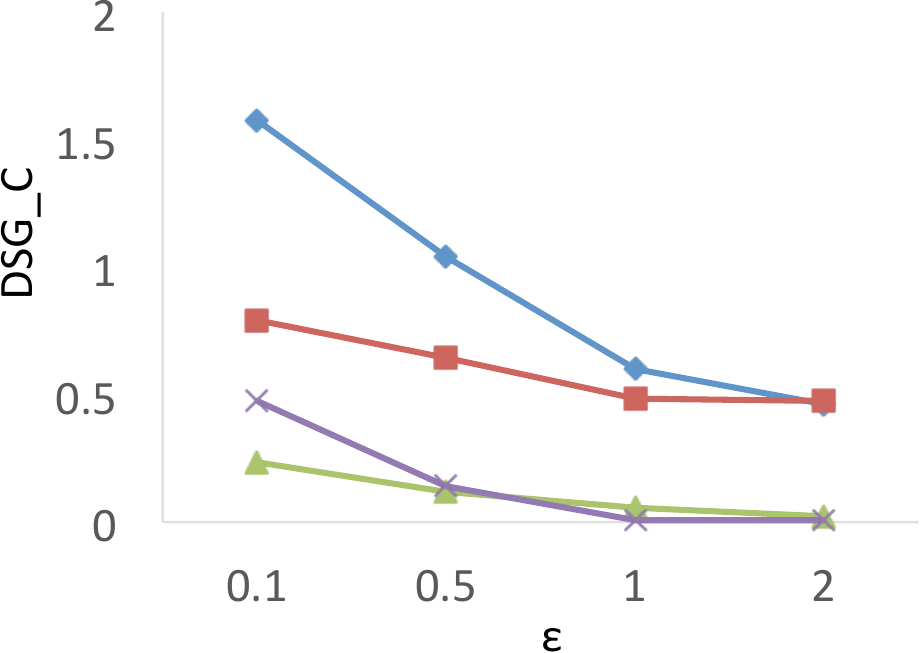} }
\subfloat[\textit{DSG$_C$-Gowalla}]{\includegraphics[scale=0.48,clip]{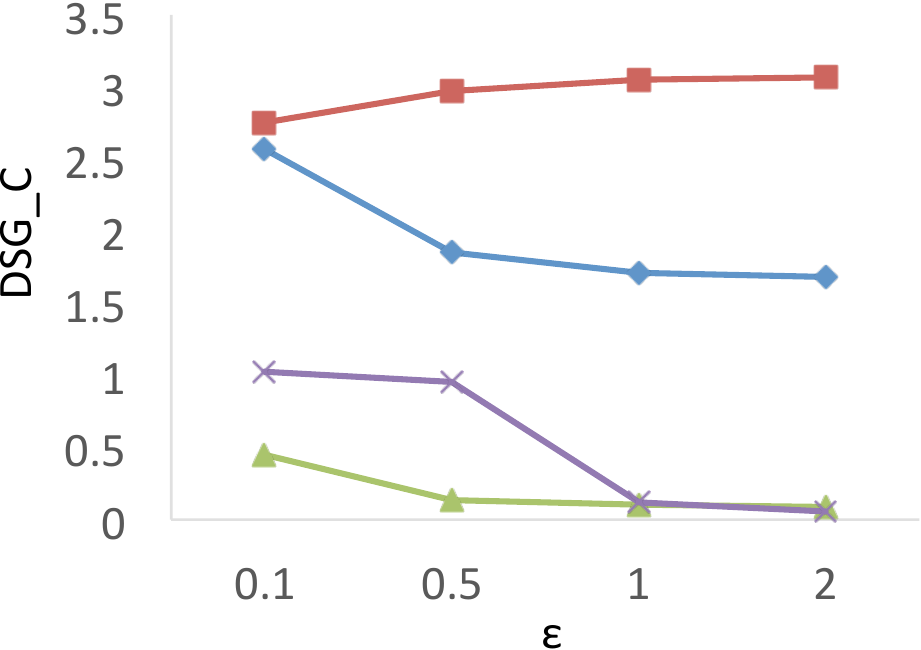} }\\
\subfloat{\includegraphics[scale=0.8,clip]{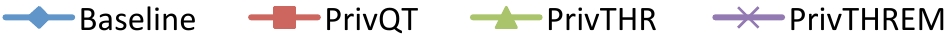} } \\
\vspace*{-1ex}
\caption{\label{fig:Dissimilarity}Comparing $DSG_C$ of 4 techniques on $DS1$, $DS2$, $DS3$ and $Gowalla$ with increasing $\epsilon$.} 
\vspace*{-4ex}
\end{figure*}

\subsection{Comparing Private $k'$ With True $k$}
We first measure the differences between the true $k$ and private $k'$s on each dataset with $\epsilon$ ranging from 0.1 to 2.0,
and the results are shown in Figure~\ref{fig:comparisonofk}.
The results show that for all datasets, when $\epsilon \geq 0.5$, 
the relative errors of $k'$, i.e., $\frac{|k'-k|}{k}$, in \privqt and \privthrem are less than 4.7\% on average,
while the relative errors of $k'$ in \baseline and \privqt range from 32.2\% to 150.5\%.
For example, in $DS2$, the true $k$ is 144.
When $\epsilon$ is 1, 
the average private $k'$ is 141.0 ($2.1\%$) for \privthr and 142.8 ($0.8\%$) for \privthrem,
while \baseline and \privqt obtain 284.0 ($97.2\%$) and 249.2 ($73.1\%$) for the average $k'$ respectively.
Note that $|Z|$ is 241 in $DS2$, 
and the difference between the average $k'$ and $k$ is 105.2 for \privqt,
which is quite close to the theoretical bound $(1-p)\frac{|Z|}{2} = 108.45$
derived from our utility analysis in Section 5.2.1.
When $\epsilon$ is 0.1, the $k'$ in \privthrem deviates from $k$ more significantly than the $k'$ in \privthr,
indicating that \privthrem is more sensitive to $\epsilon$ than \privthr as discussed in Section~\ref{subsubsec:privthrem}.
For example, in $DS2$, the average $k'$ in \privthrem is 82.8 ($42.5\%$) while the average $k'$ in \privthr is 131.2 ($8.9\%$).

\subsection{Results of $DSG_C$}

Figure~\ref{fig:Dissimilarity} shows the results of $DSG_C$ for the four techniques when the privacy budget ranges from 0.1 to 2.0.  
X-axis shows the privacy budgets, and Y-axis denotes the values of $DSG_C$.
\eat{
We found that in quite a lot of cases, $DSG_C$ is larger than $DSG$ when the setting is exactly the same. 
The reason is that $DSG_C$ might double-count the dissimilar significant grids between $T$ and $P$.
For example, suppose there are two clusters in $T$, cluster $C_1$ includes significant grids $\{g_1, g_2\}$ and cluster $C_2$ contains grid $\{g_3\}$. 
In $P$, cluster $C_1$ includes significant grid $\{g_1\}$ while cluster $C_2$ contains grids $\{g_2, g_3\}$. 
According to the definition of $DSG_C$, we found $g_2$ is counted twice in $DSG_C$. 
This double-counting causes higher $DSG_C$ value when the setting is the same. }
As shown in the results, both \privthr and \privthrem achieve smaller $DSG_C$ values than \baseline and \privqt on all four datasets for all budgets.
The reason is that though the noisy significant grids generated by \baseline and \privqt may be similar to the true significant grids,
these noisy significant grids result in very different shapes of clusters and thus result in a large value of $DSG_C$,
while \privthr and \privthrem preserves more accurate cluster shapes.
\eat{
The reason is that $DSG_C$ considers the differences of both significant grids and clusters.
Therefore even if the noisy significant grids are similar to the true significant grids, these noisy significant grids may result in very different shapes of clusters and thus result in a large value of $DSG_C$.}
For example, in $DS3$, the narrow line-shape areas and the gap between two adjacent clusters are sensitive to noise.
If some noisy significant grids appear in these areas, two clusters may be merged into one; if some significant grids disappear due to noise, one cluster might be split into two clusters.
Such changes cause $DSG_C$ to increase significantly.

Unlike the other techniques, \privqt benefits little from the increased privacy budgets.
For \privqt, the difference between $k'$ and $k$ in \privqt is dominated by $\frac{|Z|}{2}$.
Increasing privacy budgets can only reduce noise magnitude
and cannot smooth such difference.

\begin{figure*}
\centering
\subfloat[\textit{F-Measure-DS1}]{\includegraphics[scale=0.48,clip]{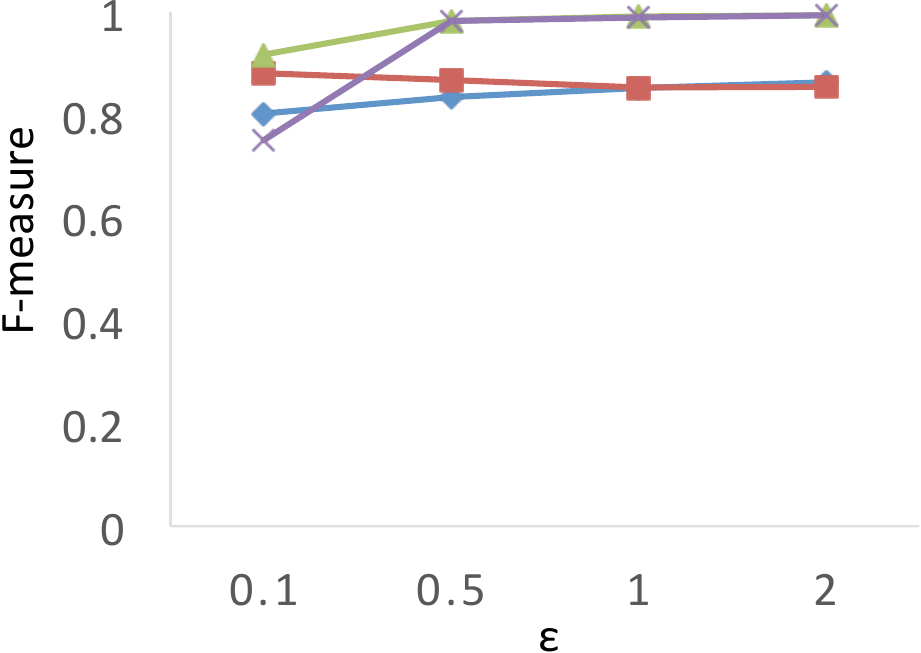} }
\subfloat[\textit{F-Measure-DS2}]{\includegraphics[scale=0.48,clip]{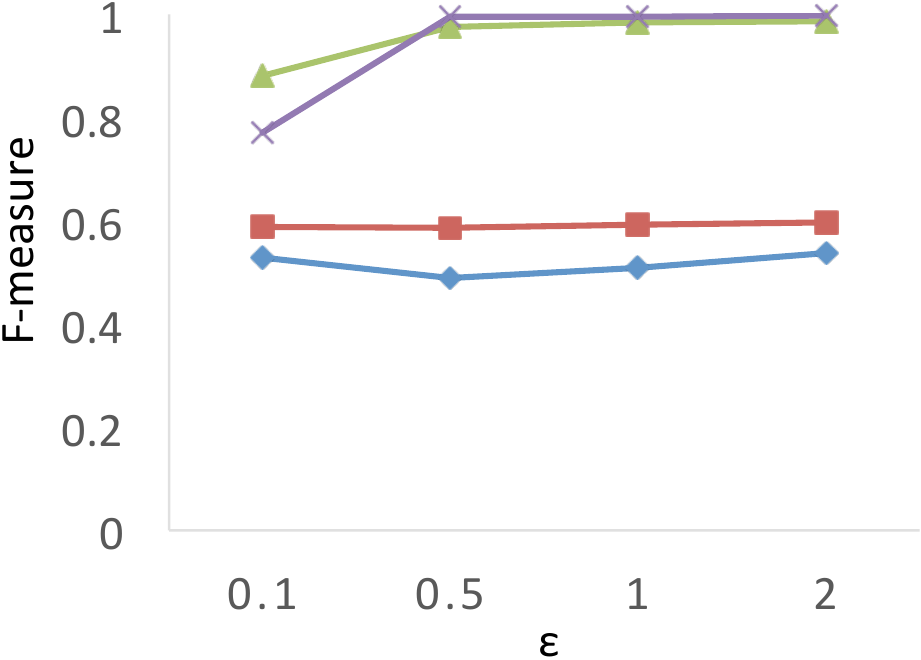} }
\subfloat[\textit{F-Measure-DS3}]{\includegraphics[scale=0.48,clip]{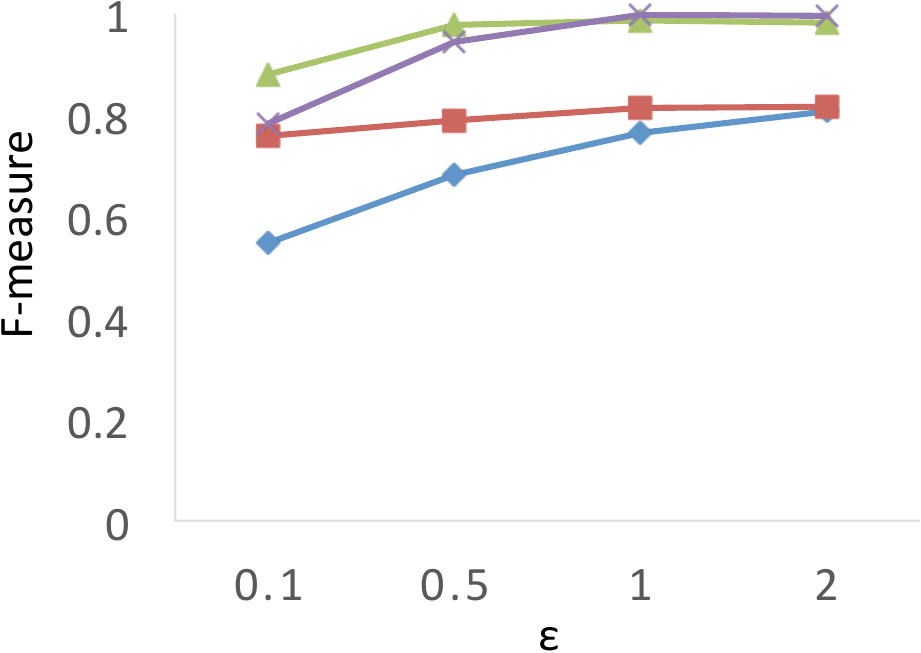} }
\subfloat[\textit{F-Measure-Gowalla}]{\includegraphics[scale=0.48,clip]{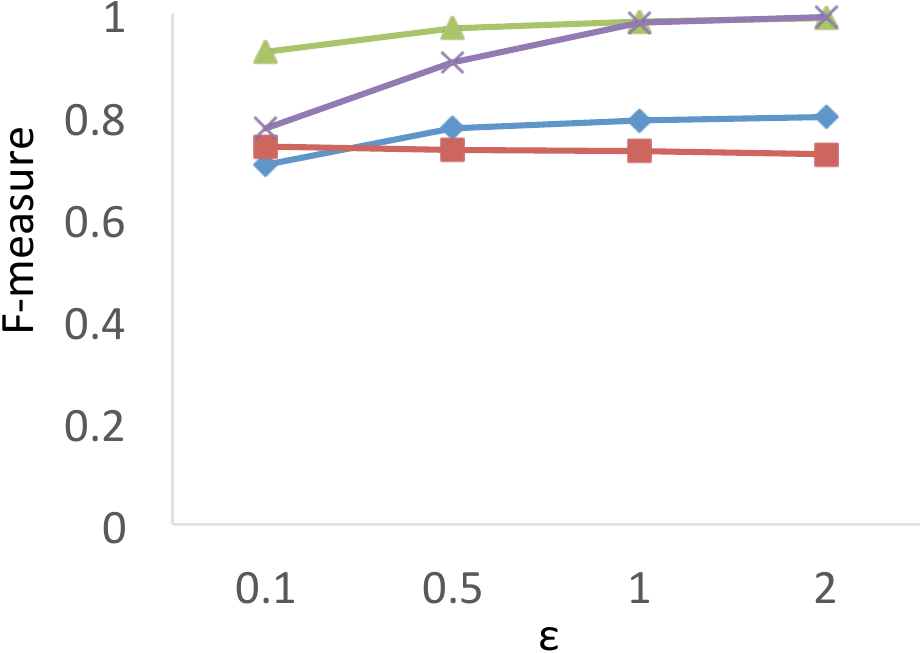} }\\
\subfloat{\includegraphics[scale=0.8,clip]{legend-new.pdf} } \\
\caption{\label{fig:fmeasure}Comparing F-Measure of 4 techniques on $DS1$, $DS2$, $DS3$ and $Gowalla$ with increasing $\epsilon$.}  
\vspace*{-4ex}
\end{figure*}

\textbf{Comparison to F-Measure Results.}
Clustering analysis usually uses F-measure as a representative external validations to measure the similarity between the ground truth (known class labels) and the clustering results~\cite{Fmeasure}. 
In our experiments, we consider the true WaveCluster results as the ground truth,
and the results of F-measure are shown in Figure~\ref{fig:fmeasure}.
The results show that \privqt and \baseline achieve high F-measure scores (more than 0.8) for almost all budgets in $DS1$, even though the private results produced by \privqt and \baseline are quite different from the true results.
For example, when $\epsilon = 0.1$, the private results of \privqt and \baseline have more than 30 clusters while the true results have only 15 clusters.
On the contrary, Figure~\ref{fig:Dissimilarity} (a) shows that $DSG_C$ is able to clearly differentiate the performances of the four techniques.
The reason is that unlike $DSG_C$ that allows only one-to-one mapping between true and private clusters, F-measure allows one-to-many or many-to-one mapping between true and private clusters.
If the size of true clusters is larger than that of private clusters, F-measure allows many to one mapping, and vice versa.
Thus, $DSG_C$ presents more strict evaluation than F-measure in computing similarity/dissimilarity.

\begin{figure*}
\vspace*{-2ex}
\centering
\subfloat[\textit{OCM-DS1}]{\includegraphics[scale=0.48,clip]{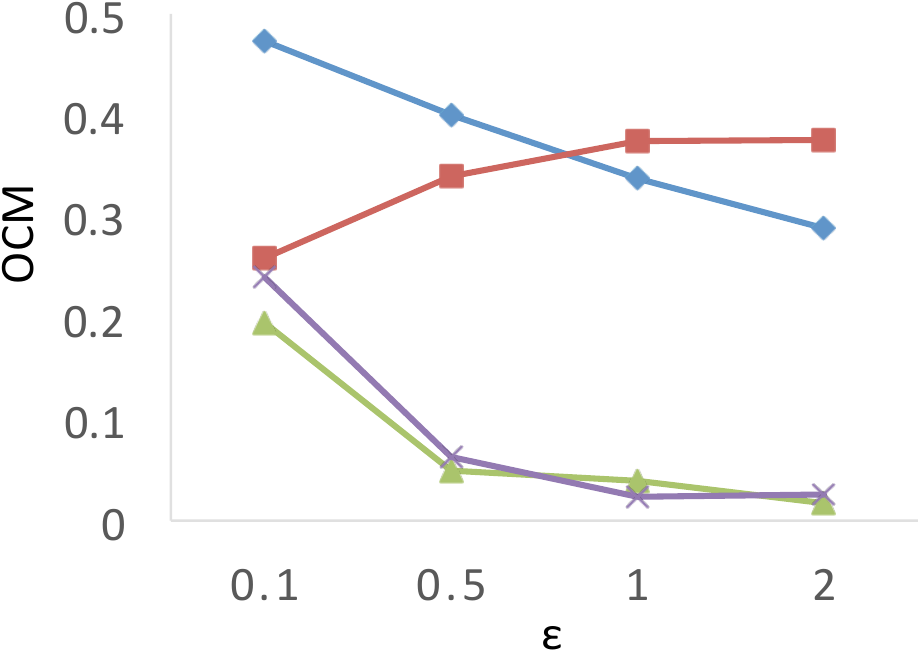} }
\subfloat[\textit{OCM-DS2}]{\includegraphics[scale=0.48,clip]{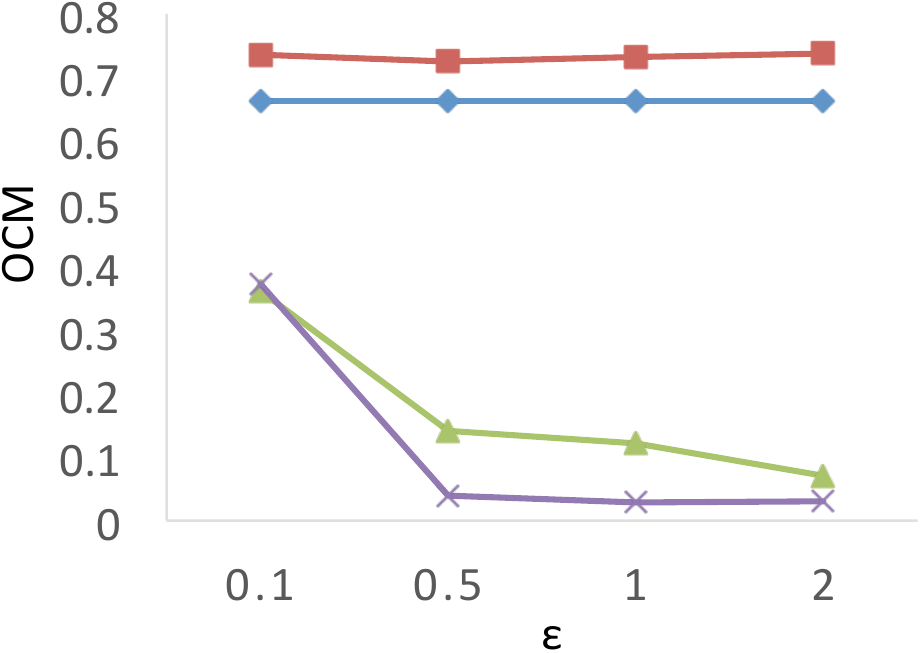} }
\subfloat[\textit{OCM-DS3}]{\includegraphics[scale=0.48,clip]{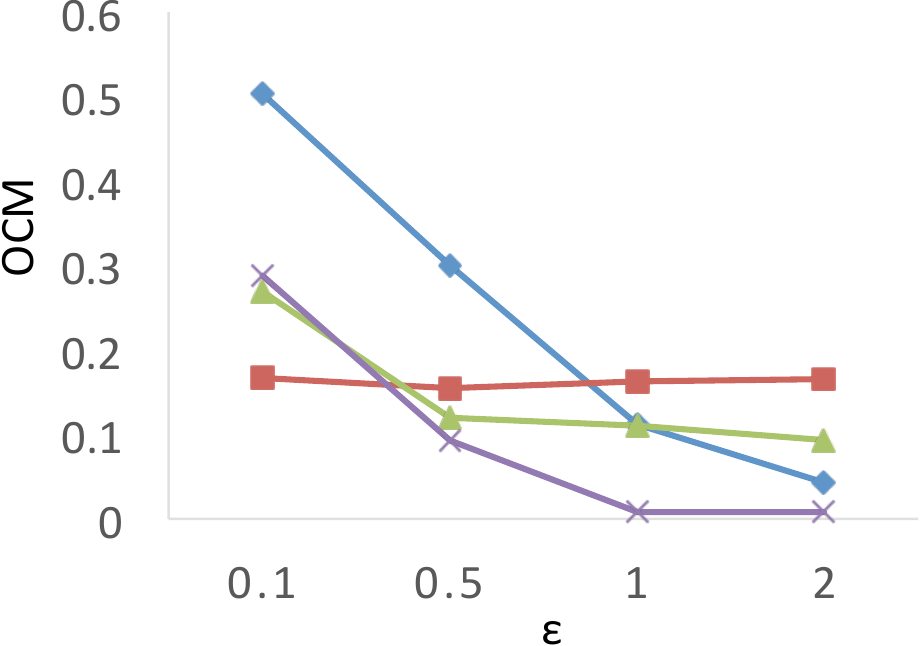} }
\subfloat[\textit{OCM-Gowalla}]{\includegraphics[scale=0.48,clip]{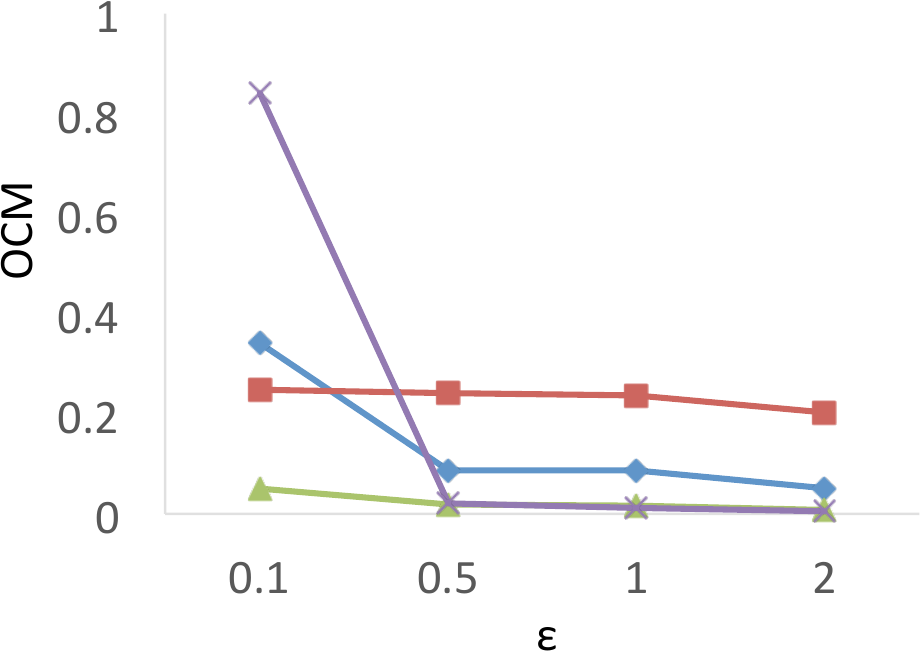} }\\
\subfloat[\textit{2CE-DS1}]{\includegraphics[scale=0.48,clip]{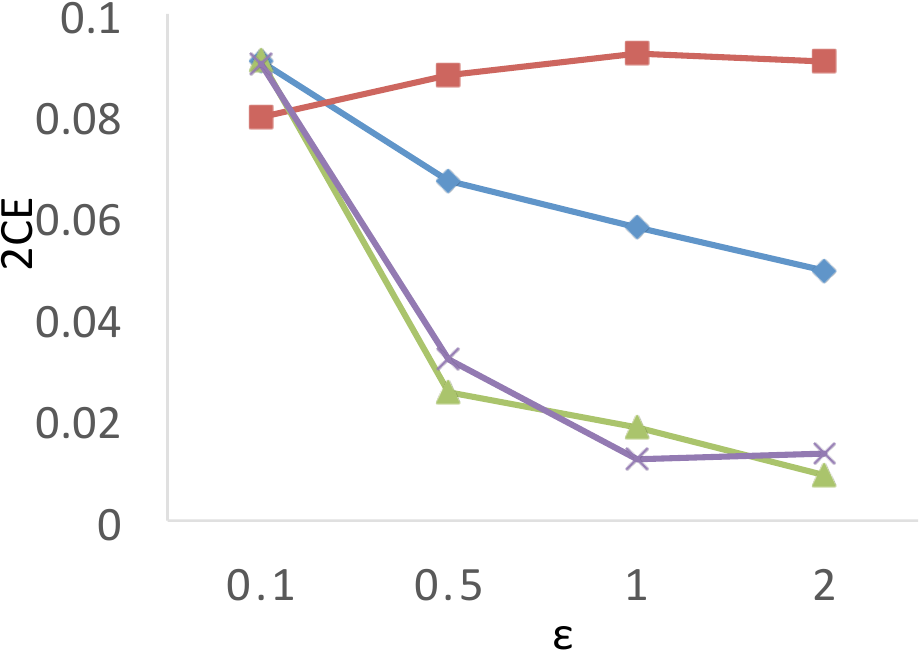} }
\subfloat[\textit{2CE-DS2}]{\includegraphics[scale=0.48,clip]{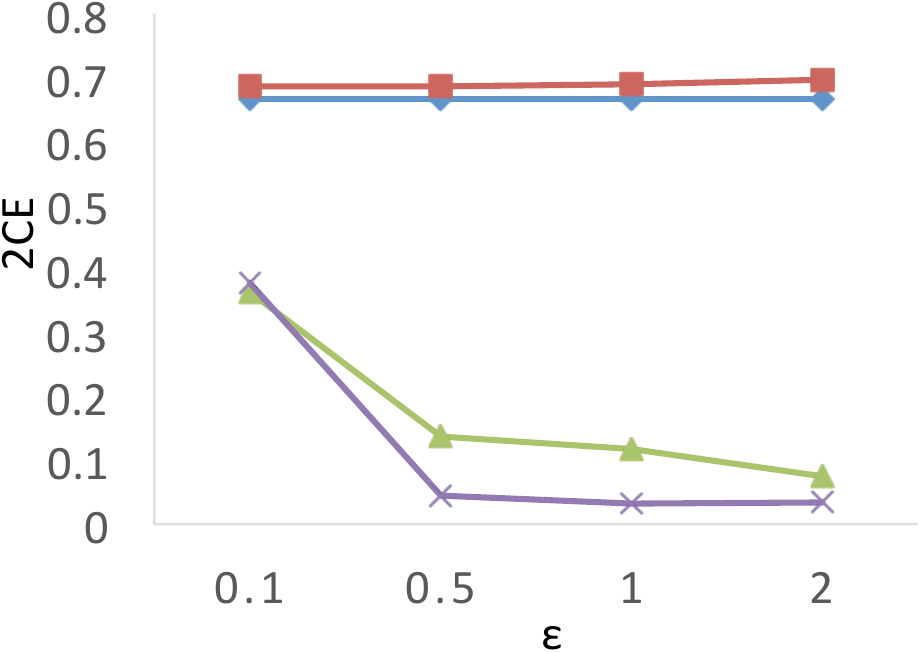} }
\subfloat[\textit{2CE-DS3}]{\includegraphics[scale=0.48,clip]{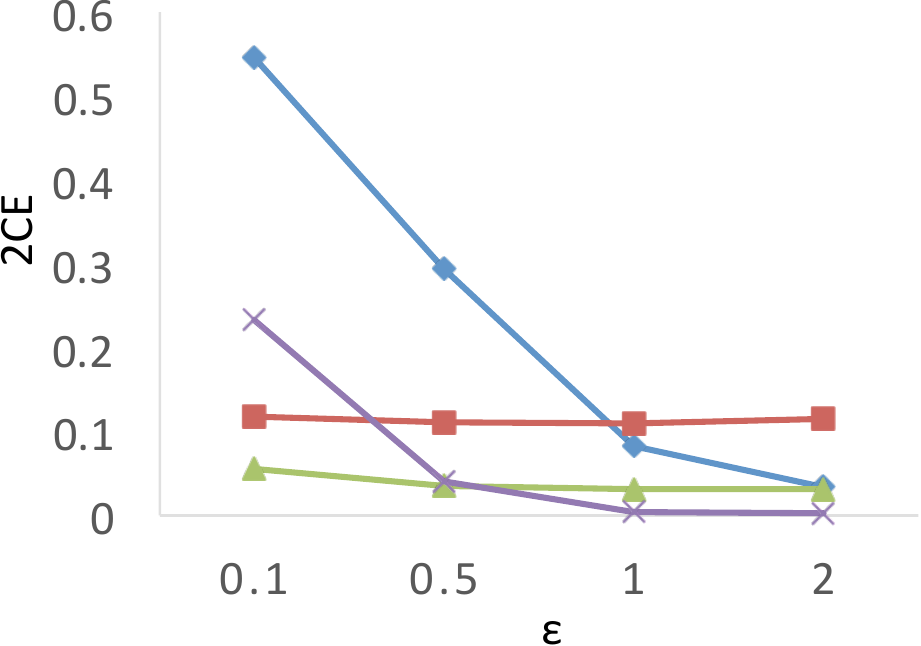} }
\subfloat[\textit{2CE-Gowalla}]{\includegraphics[scale=0.48,clip]{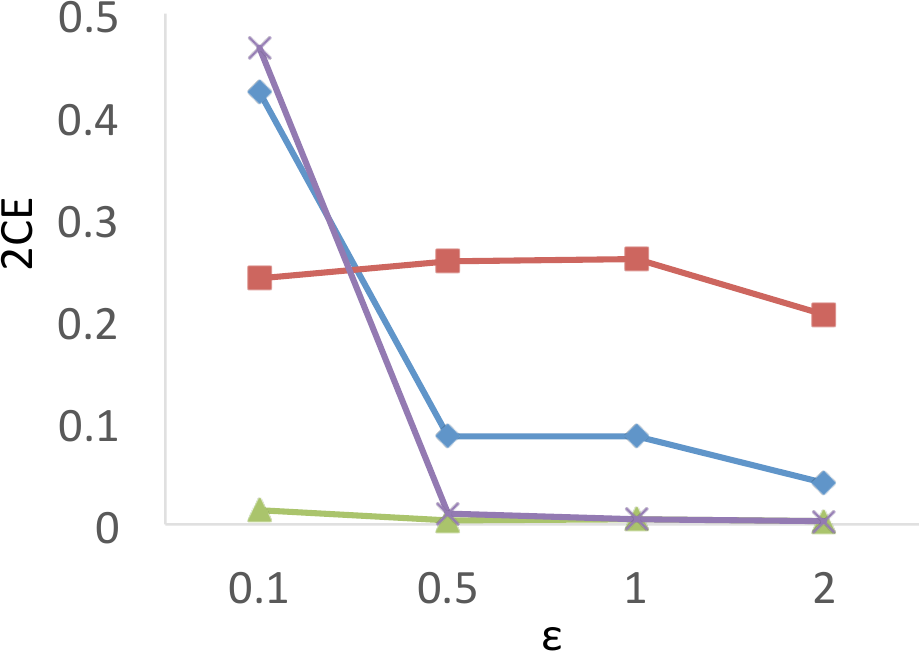} }\\
\subfloat{\includegraphics[scale=0.8,clip]{legend-new.pdf} } \\
\vspace*{-1ex}
\caption{\label{fig:Dissimilarity2}Comparing $OCM$ and $2CE$ of 4 techniques on $DS1$, $DS2$, $DS3$ and $Gowalla$ with increasing $\epsilon$.}  
\vspace*{-1ex}
\end{figure*}

\subsection{Results of $OCM$ and $2CE$}

\textbf{Results of $OCM$.}
Figure~\ref{fig:Dissimilarity2} shows the results of $OCM$ for the four techniques.
X-axis denotes the privacy budgets while Y-axis denotes the values of $OCM$.
As shown in the results, \privthr and \privthrem achieve smaller $OCM$ values than \baseline and \privqt for all datasets when $\epsilon$ ranges from 0.5 to 2.0.
When $\epsilon$ is greater than 0.5, the $OCM$ values of \privthr and \privthrem are less than 0.15 on $DS1$, $DS3$, and $Gowalla$, indicating the private classifier $clf_p$ maintains highly similar prediction results as the true classifier $clf_t$. 
On $DS2$ that contains 3 spirals, \privthrem still maintains a very low $OCM$ value ($<$ 0.1) when $\epsilon$ is greater than 0.5 while \privthr has a slightly worse $OCM$ value (ranging from 0.1 to 0.2).
Such results show that \privthrem is more resilient to noise for concave-shaped data than \privthr.

\textbf{Results of $2CE$.}
Figure~\ref{fig:Dissimilarity2} shows the results of $2CE$ for the four techniques.
X-axis denotes the privacy budgets while Y-axis denotes the values of $2CE$.
As shown in the results, \privthr and \privthrem achieve smaller $2CE$ values than \baseline and \privqt for all datasets when $\epsilon$ ranges from 0.5 to 2.0.

In general, all four techniques exhibit similar trends of $2CE$ as their trends in $OCM$.
On $DS1$, all four techniques have very low $2CE$ values ($<$ 0.1) though their corresponding $OCM$ values are much higher (ranging from 0.05 to 0.5).
The reason is that $2CE$ captures the relationships between data points while $OCM$ focuses on the mappings of classes.
If there are $k$ test samples out of $N$ total samples having different prediction results in the true and private results, $2CE$ expresses the differences as $C(k,2) + k(N-k)$ over the total combinations of test samples $C(N,2)$, while $OCM$ expresses the differences as $k$ over $N$.
On $DS1$, the $k$ test samples are predicted to be in the same cluster in the private results and $C(k,2)$ becomes close to 0.
In this case, only $k(N-k)$ matters in the computation of $2CE$.
Given that $C(N,2)$ is much larger than $N$ and $k(N-k)$ when $N$ of $DS1$ is about 30,000, $2CE$ has a smaller value than $OCM$ for measuring the differences, and thus is less sensitive to the noise on $DS1$.

\textbf{Budget Allocation for \privthr.}
Based on the utility analysis Section 5.2.2, $\epsilon_1$ for private quantization affects the accuracy of $\gamma$, and $\epsilon_2$ for obtaining $|Z|'$ affects the accuracy of $\beta$.
As the constant factor of $\gamma$, $\frac{8}{\epsilon_1}\ln{(\frac{4(|L|+|Z|)}{\omega})}$, is larger than the constant factor of $\beta$, $\frac{2}{\epsilon_2}\ln{(\frac{1}{\omega})}$,
more budget should be allocated for $\epsilon_1$ to achieve better utility.
We evaluate the values of $DSG_C$ of \privthr on $DS1$ under different budget allocation strategies,
ranging from 1\% for $\epsilon_1$ to 99\% for $\epsilon_1$.
Based on the results, the budget allocation strategy with 90\% for $\epsilon_1$  and 10\% for $\epsilon_2$ performs the best.
The results of other measures on $DS1$ show the similar results, 
and the results of all the two types of measures on other datasets also show the similar results.
Detailed results are omitted.

\section{Conclusion}\label{section-conclusion}
In this paper we have addressed the problem of cluster analysis with differential privacy. We take a well-known effective and efficient clusteing algorithm called WaveCluster, and propose several ways to introduce randomness in the computation of WaveCluster. We also devise several new quantitative measures for examining the dissimilarity between the non-private and differentially private results and the usefulness of differentially private results in classification. In the future, we will investigate under differential privacy other categories of clustering algorithms, such as hierarchical clustering. Another important problem is to explore the applicability of differentially private clustering in those cases where the users do not have good knowledge about the dataset, and the parameters of the algorithms should be inferred in a differentially private way.

\vspace*{1.5ex}
\textbf{Acknowledgments.} This work is supported in part by the National Science Foundation under the awards CNS-1314229.

\balance
\bibliographystyle{abbrv}




%
%
%
%


\end{document}